\documentclass[
12pt,
draftclsnofoot,
onecolumn,
]{IEEEtran}
\usepackage{graphicx}
\usepackage{subfigure}
\usepackage{amsmath}
\usepackage{amssymb}
\usepackage{amsthm}
\usepackage{afterpage}
\usepackage{verbatim}
\usepackage{psfrag}
\usepackage{color}
\usepackage{cite}
\usepackage{setspace}
\usepackage{epsfig}
\usepackage{footmisc}
\usepackage{fmtcount}
\usepackage{stackrel}
\usepackage{float}
\usepackage{breqn}
\usepackage{enumerate}
\usepackage{rotating}
\usepackage{url}
\usepackage{tabularx}
\usepackage{tikz}
\usepackage{multirow}
\usepackage{makecell}

\usepackage[nolist]{acronym}
\usepackage{esint} 
\usepackage{multirow}

\usepackage[per-mode=symbol]{siunitx}
\theoremstyle{plain}

\theoremstyle{plain}
\newtheorem{theorem}{Theorem}
\theoremstyle{plain}
\theoremstyle{plain}

\newtheorem{lemma}{Lemma}

\theoremstyle{remark}

\theoremstyle{remark}

\theoremstyle{plain}

\theoremstyle{plain}

\theoremstyle{plain}

  	\newtheorem{innercustomgeneric}{\customgenericname}
\providecommand{\customgenericname}{}
\newcommand{\newcustomtheorem}[2]{%
  \newenvironment{#1}[1]
  {%
   \renewcommand\customgenericname{#2}%
   \renewcommand\theinnercustomgeneric{##1}%
   \innercustomgeneric
  }
  {\endinnercustomgeneric}
}

\newcustomtheorem{customthm}{Theorem}
\newcustomtheorem{customlemma}{Lemma}
 \newcustomtheorem{customcorollary}{Corollary}

\usepackage[capitalise]{cleveref}
\Crefname{equation}{Eq.\!}{Eqs.\!}
\Crefname{figure}{Fig.\!}{Figs.\!}
\Crefname{tabular}{Tab.\!}{Tabs.\!}
\Crefname{section}{Sec.\!}{Secs.\!}

 \usepackage{pgfplots}
 \usepackage{pgfplotstable}
 \usepgfplotslibrary{patchplots}
 
  \allowdisplaybreaks 

 \usepackage{setspace}	

\begin{document}
 	\pgfplotsset{every axis/.append style={
 			line width=1pt,
 			legend style={font=\large, at={(0.97,0.85)}}},
 	} %
\begin{acronym}

\acro{5G-NR}{5G New Radio}
\acro{3GPP}{3rd Generation Partnership Project}
\acro{AC}{address coding}
\acro{ACF}{autocorrelation function}
\acro{ACR}{autocorrelation receiver}
\acro{ADC}{analog-to-digital converter}
\acrodef{aic}[AIC]{Analog-to-Information Converter}     
\acro{AIC}[AIC]{Akaike information criterion}
\acro{aric}[ARIC]{asymmetric restricted isometry constant}
\acro{arip}[ARIP]{asymmetric restricted isometry property}

\acro{ARQ}{automatic repeat request}
\acro{AUB}{asymptotic union bound}
\acrodef{awgn}[AWGN]{Additive White Gaussian Noise}     
\acro{AWGN}{additive white Gaussian noise}

\acro{APSK}[PSK]{asymmetric PSK} 

\acro{waric}[AWRICs]{asymmetric weak restricted isometry constants}
\acro{warip}[AWRIP]{asymmetric weak restricted isometry property}
\acro{BCH}{Bose, Chaudhuri, and Hocquenghem}        
\acro{BCHC}[BCHSC]{BCH based source coding}
\acro{BEP}{bit error probability}
\acro{BFC}{block fading channel}
\acro{BG}[BG]{Bernoulli-Gaussian}
\acro{BGG}{Bernoulli-Generalized Gaussian}
\acro{BPAM}{binary pulse amplitude modulation}
\acro{BPDN}{Basis Pursuit Denoising}
\acro{BPPM}{binary pulse position modulation}
\acro{BPSK}{binary phase shift keying}
\acro{BPZF}{bandpass zonal filter}
\acro{BSC}{binary symmetric channels}              
\acro{BU}[BU]{Bernoulli-uniform}
\acro{BER}{bit error rate}
\acro{BS}{base station}

\acro{CP}{Cyclic Prefix}
\acrodef{cdf}[CDF]{cumulative distribution function}   
\acro{CDF}{cumulative distribution function}
\acrodef{c.d.f.}[CDF]{cumulative distribution function}
\acro{CCDF}{complementary cumulative distribution function}
\acrodef{ccdf}[CCDF]{complementary CDF}               
\acrodef{c.c.d.f.}[CCDF]{complementary cumulative distribution function}
\acro{CD}{cooperative diversity}

\acro{CDMA}{Code Division Multiple Access}
\acro{ch.f.}{characteristic function}
\acro{CIR}{channel impulse response}
\acro{cosamp}[CoSaMP]{compressive sampling matching pursuit}
\acro{CR}{cognitive radio}
\acro{cs}[CS]{compressed sensing}                   
\acrodef{cscapital}[CS]{Compressed sensing}
\acrodef{CS}[CS]{compressed sensing}
\acro{CSI}{channel state information}

\acro{CCSDS}{consultative committee for space data systems}
\acro{CC}{convolutional coding}

\acro{DAA}{detect and avoid}
\acro{DAB}{digital audio broadcasting}
\acro{DCT}{discrete cosine transform}
\acro{dft}[DFT]{discrete Fourier transform}
\acro{DR}{distortion-rate}
\acro{DS}{direct sequence}
\acro{DS-SS}{direct-sequence spread-spectrum}
\acro{DTR}{differential transmitted-reference}
\acro{DVB-H}{digital video broadcasting\,--\,handheld}
\acro{DVB-T}{digital video broadcasting\,--\,terrestrial}
\acro{DL}{downlink}
\acro{DSSS}{Direct Sequence Spread Spectrum}
\acro{DFT-s-OFDM}{Discrete Fourier Transform-spread-Orthogonal Frequency Division Multiplexing}
\acro{DAS}{distributed antenna system}
\acro{DNA}{Deoxyribonucleic Acid}

\acro{EC}{European Commission}
\acro{EED}[EED]{exact eigenvalues distribution}
\acro{EIRP}{Equivalent Isotropically Radiated Power}
\acro{ELP}{equivalent low-pass}
\acro{eMBB}{Enhanced Mobile Broadband}
\acro{EMF}{electric and magnetic fields}
\acro{EU}{European union}

\acro{FC}[FC]{fusion center}
\acro{FCC}{Federal Communications Commission}
\acro{FEC}{forward error correction}
\acro{FFT}{fast Fourier transform}
\acro{FH}{frequency-hopping}
\acro{FH-SS}{frequency-hopping spread-spectrum}
\acrodef{FS}{Frame synchronization}
\acro{FSsmall}[FS]{frame synchronization}  
\acro{FDMA}{Frequency Division Multiple Access}    

\acro{gNB}{generation node B base station}

\acro{GA}{Gaussian approximation}
\acro{GF}{Galois field }
\acro{GG}{Generalized-Gaussian}
\acro{GIC}[GIC]{generalized information criterion}
\acro{GLRT}{generalized likelihood ratio test}
\acro{GPS}{Global Positioning System}
\acro{GMSK}{Gaussian minimum shift keying}
\acro{GSMA}{Global System for Mobile communications Association}

\acro{HAP}{high altitude platform}

\acro{IDR}{information distortion-rate}
\acro{IFFT}{inverse fast Fourier transform}
\acro{iht}[IHT]{iterative hard thresholding}
\acro{i.i.d.}{independent, identically distributed}
\acro{IoT}{Internet of Things}                      
\acro{IR}{impulse radio}
\acro{lric}[LRIC]{lower restricted isometry constant}
\acro{lrict}[LRICt]{lower restricted isometry constant threshold}
\acro{ISI}{intersymbol interference}
\acro{ITU}{International Telecommunication Union}
\acro{ICNIRP}{International Commission on Non-Ionizing Radiation Protection}
\acro{IEEE}{Institute of Electrical and Electronics Engineers}
\acro{ICES}{IEEE international committee on electromagnetic safety}
\acro{IEC}{International Electrotechnical Commission}
\acro{IARC}{International Agency on Research on Cancer}
\acro{IS-95}{Interim Standard 95}

\acro{KPI}{Key Performance Indicator}

\acro{LEO}{low earth orbit}
\acro{LF}{likelihood function}
\acro{LLF}{log-likelihood function}
\acro{LLR}{log-likelihood ratio}
\acro{LLRT}{log-likelihood ratio test}
\acro{LOS}{Line-of-Sight}
\acro{LRT}{likelihood ratio test}
\acro{wlric}[LWRIC]{lower weak restricted isometry constant}
\acro{wlrict}[LWRICt]{LWRIC threshold}
\acro{LPWAN}{low power wide area network}
\acro{LoRaWAN}{Low power long Range Wide Area Network}

\acro{MB}{multiband}
\acro{MC}{multicarrier}
\acro{MDS}{mixed distributed source}
\acro{MF}{matched filter}
\acro{m.g.f.}{moment generating function}
\acro{MI}{mutual information}
\acro{MIMO}{multiple-input multiple-output}
\acro{MISO}{multiple-input single-output}
\acrodef{maxs}[MJSO]{maximum joint support cardinality}                       
\acro{ML}[ML]{maximum likelihood}
\acro{MMSE}{minimum mean-square error}
\acro{MMV}{multiple measurement vectors}
\acrodef{MOS}{model order selection}
\acro{M-PSK}[${M}$-PSK]{$M$-ary phase shift keying}                       
\acro{M-APSK}[${M}$-PSK]{$M$-ary asymmetric PSK} 

\acro{M-QAM}[$M$-QAM]{$M$-ary quadrature amplitude modulation}
\acro{MRC}{maximal ratio combiner}                  
\acro{maxs}[MSO]{maximum sparsity order}                                      
\acro{M2M}{machine to machine}                                                
\acro{MUI}{multi-user interference}
\acro{mMTC}{massive Machine Type Communications}      
\acro{mm-Wave}[mm-Wave]{millimeter-wave}
\acro{MP}{mobile phone}
\acro{MPE}{maximum permissible exposure}
\acro{MAC}{media access control}
\acro{NB}{narrowband}
\acro{NBI}{narrowband interference}
\acro{NLA}{nonlinear sparse approximation}
\acro{NLOS}{Non-Line of Sight}
\acro{NTIA}{National Telecommunications and Information Administration}
\acro{NTP}{National Toxicology Program}

\acro{OC}{optimum combining}                             
\acro{OC}{optimum combining}
\acro{ODE}{operational distortion-energy}
\acro{ODR}{operational distortion-rate}
\acro{OFDM}{orthogonal frequency-division multiplexing}
\acro{omp}[OMP]{orthogonal matching pursuit}
\acro{OSMP}[OSMP]{orthogonal subspace matching pursuit}
\acro{OQAM}{offset quadrature amplitude modulation}
\acro{OQPSK}{offset QPSK}
\acro{OFDMA}{Orthogonal Frequency-division Multiple Access}

\acro{OQPSK/PM}{OQPSK with phase modulation}

\acro{PAM}{pulse amplitude modulation}
\acro{PAR}{peak-to-average ratio}
\acrodef{pdf}[PDF]{probability density function}                      
\acro{PDF}{probability density function}
\acrodef{p.d.f.}[PDF]{probability distribution function}
\acro{PDP}{power dispersion profile}
\acro{PMF}{probability mass function}                             
\acrodef{p.m.f.}[PMF]{probability mass function}
\acro{PN}{pseudo-noise}
\acro{PPM}{pulse position modulation}
\acro{PRake}{Partial Rake}
\acro{PSD}{power spectral density}
\acro{PSEP}{pairwise synchronization error probability}
\acro{PSK}{phase shift keying}
\acro{PD}{power density}
\acro{8-PSK}[$8$-PSK]{$8$-phase shift keying}
\acro{PHP}{Poisson hole process}
 \acro{PPP}{Poisson point process}
\acrodefplural{PPP}{Poisson point processes}
\acro{FSK}{frequency shift keying}

\acro{QAM}{Quadrature Amplitude Modulation}
\acro{QPSK}{quadrature phase shift keying}
\acro{OQPSK/PM}{OQPSK with phase modulator }

\acro{RD}[RD]{raw data}
\acro{RDL}{"random data limit"}
\acro{ric}[RIC]{restricted isometry constant}
\acro{rict}[RICt]{restricted isometry constant threshold}
\acro{rip}[RIP]{restricted isometry property}
\acro{ROC}{receiver operating characteristic}
\acro{rq}[RQ]{Raleigh quotient}
\acro{RS}[RS]{Reed-Solomon}
\acro{RSC}[RSSC]{RS based source coding}
\acro{r.v.}{random variable}                               
\acro{R.V.}{random vector}
\acro{RMS}{root mean square}
\acro{RFR}{radiofrequency radiation}
\acro{RIS}{Reconfigurable intelligent surface}

\acro{SA}[SA-Music]{subspace-augmented MUSIC with OSMP}
\acro{SCBSES}[SCBSES]{Source Compression Based Syndrome Encoding Scheme}
\acro{SCM}{sample covariance matrix}
\acro{SEP}{symbol error probability}
\acro{SG}[SG]{sparse-land Gaussian model}
\acro{SIMO}{single-input multiple-output}
\acro{SINR}{signal-to-interference plus noise ratio}
\acro{SIR}{signal-to-interference ratio}
\acro{SISO}{single-input single-output}
\acro{SMV}{single measurement vector}
\acro{SNR}[\textrm{SNR}]{signal-to-noise ratio} 
\acro{sp}[SP]{subspace pursuit}
\acro{SS}{spread spectrum}
\acro{SW}{sync word}
\acro{SAR}{specific absorption rate}
\acro{SSB}{synchronization signal block}
\acro{SDG}{Sustainable Development Goal}

\acro{TH}{time-hopping}
\acro{ToA}{time-of-arrival}
\acro{TR}{transmitted-reference}
\acro{TW}{Tracy-Widom}
\acro{TWDT}{TW Distribution Tail}
\acro{TCM}{trellis coded modulation}
\acro{TDD}{time-division duplexing}
\acro{TDMA}{Time Division Multiple Access}

\acro{UAV}{unmanned aerial vehicle}
\acro{uric}[URIC]{upper restricted isometry constant}
\acro{urict}[URICt]{upper restricted isometry constant threshold}
\acro{UWB}{ultrawide band}
\acro{UWBcap}[UWB]{Ultrawide band}   
\acro{URLLC}{Ultra Reliable Low Latency Communications}
         
\acro{wuric}[UWRIC]{upper weak restricted isometry constant}
\acro{wurict}[UWRICt]{UWRIC threshold}                
\acro{UE}{user equipment}
\acro{UL}{uplink}

\acro{WiM}[WiM]{weigh-in-motion}
\acro{WLAN}{wireless local area network}

\acro{wm}[WM]{Wishart matrix}                               
\acroplural{wm}[WM]{Wishart matrices}
\acro{WMAN}{wireless metropolitan area network}
\acro{WPAN}{wireless personal area network}
\acro{wric}[WRIC]{weak restricted isometry constant}
\acro{wrict}[WRICt]{weak restricted isometry constant thresholds}
\acro{wrip}[WRIP]{weak restricted isometry property}
\acro{WSN}{wireless sensor network}                        
\acro{WSS}{wide-sense stationary}
\acro{WHO}{World Health Organization}
\acro{WP}{work package}

\acro{sss}[SpaSoSEnc]{sparse source syndrome encoding}
\acro{SO}{strategic objective}

\acro{VLC}{visible light communication}
\acro{RF}{radio frequency}
\acro{FSO}{free space optics}
\acro{IoST}{Internet of space things}

\acro{GSM}{Global System for Mobile Communications}
\acro{2G}{second-generation cellular network}
\acro{3G}{third-generation cellular network}
\acro{4G}{fourth-generation cellular network}
\acro{5G}{5th-generation cellular network}	
\acro{gNB}{next generation node B base station}
\acro{NR}{New Radio}
\acro{UN}{United Nations}
\acro{UMTS}{Universal Mobile Telecommunications Service}
\acro{LTE}{Long Term Evolution}

\acro{QoS}{quality of service}
\end{acronym}

\newcommand{\SAR} {\mathrm{SAR}}
\newcommand{\WBSAR} {\mathrm{SAR}_{\mathsf{WB}}}
\newcommand{\gSAR} {\mathrm{SAR}_{10\si{\gram}}}
\newcommand{\Sab} {S_{\mathsf{ab}}}
\newcommand{\Eavg} {E_{\mathsf{avg}}}
\newcommand{\ft}{f_{\textsf{th}}}
\newcommand{\alphatf}{\alpha_{24}}
\newcommand{\xvi}{{\mathbf{X}}_{i}}
\newcommand{\x}{{\mathbf{x}}}

%
\title{Joint Uplink and Downlink {EMF} Exposure: Performance Analysis and Design Insights  
%
}
\author{Lin Chen, Ahmed~Elzanaty,~\IEEEmembership{Senior Member,~IEEE}, Mustafa~A.~Kishk,~\IEEEmembership{Member,~IEEE}, Luca~Chiaraviglio,~\IEEEmembership{Senior Member,~IEEE},  and Mohamed-Slim Alouini,~\IEEEmembership{Fellow,~IEEE} 
\thanks{Lin Chen is with the Department of Information Engineering, The Chinese University of Hong Kong (CUHK), Hong Kong (e-mail: lin.chen@link.cuhk.edu.hk). }
\thanks{A. Elzanaty is with the 5GIC \& 6GIC, Institute for Communication Systems (ICS), University of Surrey, Guildford, GU2 7XH, United Kingdom (email: a.elzanaty@surrey.ac.uk).}
\thanks{M. A. Kishk is with the Department of Electronic Engineering, Maynooth University, Maynooth, W23 F2H6, Ireland (email: mustafa.kishk@mu.ie).}
\thanks{L. Chiaraviglio is with the Department of Electronic Engineering, Universita degli Studi di Roma Tor Vergata, 00133 Rome, Italy and Consorzio Nazionale Interuniversitario per le Telecomunicazioni (CNIT), Parma, Italy (email:  luca.chiaraviglio@uniroma2.it).}
\thanks{M.-S. Alouini is with KAUST, CEMSE division, Thuwal 23955-6900, Saudi Arabia (email: slim.alouini@kaust.edu.sa). }
}
\markboth{}{Elzanaty {\MakeLowercase{\textit{et al.}}}: Adaptive Coded Modulation for IM/DD Free-Space Optical Backhauling: A Probabilistic Shaping Approach}
\maketitle
%
%
%


\maketitle

\begin{abstract}
Installing more base stations (BSs) into the existing cellular infrastructure is an essential way to provide greater network capacity and higher data rate in the 5th-generation cellular networks (5G).
However, a non-negligible amount of population is concerned that such network densification will generate a notable increase in exposure to electric and magnetic fields (EMF) over the territory.
In this paper, we analyze the downlink, uplink, and joint downlink\&uplink exposure induced by the radiation from {BS}s and personal user equipment ({UE}), respectively, in terms of the received power density and exposure index. In our analysis, we consider the EMF restrictions set by the regulatory authorities such as the  minimum distance between restricted areas (e.g., schools and hospitals) and {BS}s, and the maximum permitted exposure. Exploiting tools from
stochastic geometry, mathematical expressions for the coverage probability and statistical {EMF} exposure are derived and validated. Tuning the system parameters such as the {BS} density and the minimum distance from a {BS} to restricted areas, we show a trade-off between reducing the population's exposure to EMF  and enhancing the network coverage performance. Then, we formulate optimization problems to maximize the performance of the {EMF}-aware cellular network while ensuring that the {EMF} exposure complies with the standard regulation limits with high probability. {For instance, the  exposure from BSs is two orders of magnitude less than the maximum permissible level when the density of BSs is less than $20$ $\text{BSs/km}^2$. }


\end{abstract}

%
\IEEEpeerreviewmaketitle

\section{Introduction}\label{sec:intro}
The fundamental requirements of \ac{5G} are low latency, high throughput, and wide coverage. One potential solution to accommodate \ac{5G} key performance indicators (KPIs) is to increase the number of \acp{BS}~\cite{boccardi2014five}.
The new \ac{5G} \acp{BS} inevitably act as additional radiation sources, {concerning some of the population about the increasing possibility of their exposure to \ac{EMF}.}  
Recently, human health related to the massive deployment of  \acp{BS} has raised public concerns \cite{switzerland}. 
{There is an urgent need to provide scientific analysis as we do not know if health effects (not known at present time) will be observed in the future.}

\textcolor{black}{
\ac{EMF} exposure in cellular networks mainly comes from \acp{BS} and \ac{UE}, related to passive exposure and active exposure, respectively. \acp{BS} emit high power through long distances in the downlink, imposing \ac{EMF} exposure to humans passively. The \ac{EMF} exposure from all \acp{BS} in the network is usually measured by the power density at the user, which can be easily transformed into electric strength~\cite{international1998guidelines}. 
Nevertheless, the exposure originating from \ac{UE} with low transmit power should also be taken into account due to the close distance between the user and the personal mobile device, which can be quantified by the received power density or the \ac{SAR}~\cite{ICNIRPGuidelines:18,2EMFsources}.
}
In fact, \ac{EMF} exposure associated with \ac{5G} \ac{RF} communications is considered as non-ionizing radiation that does not have enough energy to ionize the cells \cite{non-ionizing}. 
Nevertheless, the non-ionizing radiation is possible to generate heating effects in the exposed tissues, i.e., thermal effects~\cite{FosterZisBal:18,belpomme2018thermal}.
In order to guarantee that the thermal effects are below acceptable safe levels, \ac{EMF} exposure guidelines are  set such as those by 
{\ac{ITU}~\cite{itu2004guidance},
\ac{ICNIRP}~\cite{ICNIRPGuidelines:18}, and \ac{FCC}~\cite{FCC}. }
Each country has its own regulations on the safety limits to the EMF~\cite{whodata,SaudiEMF:21,chiang2009rationale,hkemf,liechtensteinemf}, mainly based on the aforementioned regulatory guidelines.
Besides these guidelines, as a further precautionary measure, some countries adopt more restricted conditions such as a minimum distance between \acp{BS} and restricted areas, e.g., schools and hospitals {\cite{hole2}}.
Besides the well-understood thermal effects, there is a debate about whether long-term exposure to \ac{RFR} may have non-thermal effects that can lead to health issues \cite{NTP:18a,WalMinKen:19}.  Therefore, accurate analysis of the exposure to \ac{EMF} is essential to permit designing \ac{EMF}-aware cellular networks.

\subsection{Related Work} \label{sec:related}
In this subsection, we discuss the most related work on the \ac{EMF} exposure, which can be divided into two categories: (i) \ac{EMF} exposure assessment and (ii) \ac{EMF}-aware network design.

{\em \ac{EMF} exposure assessment.}
The evaluation of \ac{EMF} exposure in cellular networks can be conducted from 
experimental measurement
and analytical points of view.
As for the experimental measurement,
the exposure induced by \acp{BS} and \acp{UE} 
was measured in~\cite{huang2016comparison}, which revealed that the exposure from a personal mobile device could not be ignored.
\textcolor{black}{
Considering the enabling technologies of \ac{5G} such as massive \ac{MIMO}, real-time beamforming, and high-frequency bands, \ac{5G} smartphones were used to capture the exposure level in a \ac{5G} network \cite{pawlak2019measuring}.
In \cite{chiaraviglio2021massive},  the distribution of \ac{EMF}  in \ac{5G} networks is investigated.
The authors in \cite{aerts2019situ} focused on the measurement of downlink exposure and showed that the exposure is well below the \ac{ICNIRP} reference level.}
Clearly, the experimental measurement of \ac{EMF} exposure illustrates the \ac{RFR} level under a specific cellular network, while it is not able to explore the effect of the system parameters on \ac{EMF} exposure for cellular network deployment. 

As for the analytical evaluation of \ac{EMF} exposure, authors in~\cite{chiaraviglio2021dense} considered a regular deployment of \acp{BS} under the assumption of a hexagonal mosaic territory.
Employing stochastic geometry, more realistic modeling of irregularly distributed \acp{BS} was given in~\cite{gontier2021stochastic}.
Compared with the experimental data of exposure in Brussels, Belgium, authors in~\cite{gontier2021stochastic} optimized the model parameters and 
verified the fitting effect of the proposed model.
\textcolor{black}{In \cite{al2020statistical}, 
the statistical received power at users was used to monitor the downlink exposure levels in a \ac{MIMO} system using tools from stochastic geometry. The above analytical modelling allows the prediction of the downlink \ac{EMF} exposure from \acp{BS} before the actual network deployment. Yet, the uplink exposure from mobile equipment and EMF restrictions such as the minimum distance between \acp{BS} and restricted areas have not been 
considered in those analytical models. Thus, these models do not lead to an accurate evaluation of total EMF exposure in real cellular networks.}


{\em \ac{EMF}-aware network design.}
\textcolor{black}{
Recently, several researchers proposed novel cellular architectures to reduce EMF exposure or improve coverage performance while limiting the exposure \cite{R1}.
\acp{RIS} were first proposed in \cite{chiaraviglio2021health} as a solution to create areas with reduced \ac{EMF}. In \cite{ibraiwish2021emf}, the \ac{RIS} phases were optimized to minimize the total uplink exposure of users. On the other hand, the work in \cite{R3} and \cite{R4} considered designing the \ac{RIS} phases to minimize the maximum exposure (min-max problem) with instantaneous and statistical channel state information, respectively. 
The authors in \cite{javedemf} applied probabilistic shaping to minimize the average \ac{EMF} exposure while ensuring a target throughput.
Considering the coexistence of macro-cell and small-cell \acp{BS}, the downlink EMF exposure and the coverage probability are studied in~\cite{muhammad2021stochastic}.
It is worth noting that the existing research on the \ac{EMF}-aware network design mainly aims at mitigating the downlink/uplink \ac{EMF} exposure on average and ensuring a target \ac{QoS}.
Nevertheless, the total exposure (including downlink and uplink exposure) is rarely taken into account when designing the cellular network despite the fact that people are exposed to the \ac{RFR} from both \acp{BS} and \acp{UE}~\cite{lou2021green}. 
Moreover, the regulation of restricted areas is seldom considered in the \ac{EMF}-aware network design; while this regulation is essential for the evaluation of \ac{EMF} exposure levels since the minimum distance between BSs and restricted areas can significantly reduce the exposure levels in restricted areas and can be selected by the system designers or decision makers. The effect of restricted areas on network performance has not been studied.
}

\subsection{Contributions}\label{sec:contri}
\textcolor{black}{
In this paper, we provide a novel framework to analyze the impact of \ac{EMF} exposure from both BSs and UE on the planning of a \ac{5G} cellular network from a statistical point of view.
In particular, the proposed cellular network model considers the restrictions on the \ac{EMF} exposure set by the regulatory authorities, including the maximum permitted \ac{EMF} exposure and the exclusion zones around restricted areas, e.g., hospitals and schools \cite{hole2}. 
Three scenarios are mainly studied in our model, including the downlink (passive exposure), the uplink (active exposure), and the joint downlink\&uplink (passive and active exposure) \ac{RFR}. 
Nevertheless, the inclusion of the restricted area, the maximum transmit power of the power with different channel inversion coefficients at \ac{UE}, and the combination of the uplink and downlink lead to some mathematical challenges. These challenges are handled through approximations such as \ac{PHP}, which are shown to be accurate through Monte Carlo simulations. 
The main contributions of our work are summarized as follows.
}
\textcolor{black}{
\begin{itemize}
    \item We propose a stochastic geometry model that captures the minimum distance ($R$) between a \ac{BS} and a restricted area. Specifically, restricted areas result in the \ac{PHP}-distributed \acp{BS}. Such \ac{PHP} approximation enables the following analysis of the effect of $R$ on the exposure and coverage performance. The accuracy of the approximation has been validated in simulation. 
    \item We quantify the downlink exposure from \acp{BS} and the uplink exposure from personal \ac{UE} by the received power density under the Nakagami-$m$ fading model. Correspondingly, we analyze the coverage probability in the downlink and uplink, respectively. Furthermore, considering a more practical scenario that people are exposed to radiation from both \acp{BS} and \ac{UE}, we provide the joint downlink\&uplink exposure analysis, referred to as exposure index. 
    \item We assess the compliance of cellular networks to the exposure guidelines.
    As opposed to the mean-value-based measurement of EMF exposure, our work provides the \ac{CDF} of exposure by using the Gil-Pelaez theorem.
    By comparing the $95$-th percentile of \ac{EMF} level with the standard limit defined by \ac{FCC}, we ensure that the EMF exposure level complies with guidelines with a high probability.
  \item  We design the system parameters, e.g., the density of the \acp{BS} and the minimum distance between the restricted areas and \acp{BS}, to \textit{(i)} maximize the coverage performance constrained by the $95$-th percentile of \ac{EMF} exposure in the downlink
  and to \textit{(ii)} minimize the joint downlink\&uplink exposure. Numerical results show the impact of system parameters on network performance and exposure, and there exist optimal values of system parameters with respect to the minimal total exposure.
  These analyses provide insights into the design of future networks to meet the \ac{QoS} and safety requirements.
\end{itemize}
}


The rest of the paper is structured as follows. Sec.~\ref{sec:system} introduces the system model. The downlink analysis of the statistical \ac{EMF} and SNR-based coverage probability are given in Sec.~\ref{sec:downlink}, including the design of system parameters.
Similar performance metrics in the uplink are presented in Sec.~\ref{sec:uplink}.
{Sec.~\ref{sec:DLandUL}} adopts exposure index to consider the total influence of both uplink exposure and downlink exposure on the population. Then, the simulation results are shown and discussed in Sec.~\ref{sec:simu}. Finally, Sec.~\ref{sec:conc} concludes the paper.

\begin{table*}[t!]\caption{Table of notations}
\centering
\begin{center}
\resizebox{\textwidth}{!}{
\renewcommand{\arraystretch}{1.4}
    \begin{tabular}{ {c} | {c} }
    \hline
        \hline
    \textbf{Notation} & \textbf{Description} \\ \hline
    ${\Psi_b}$ & The \ac{PPP} modeling the baseline locations of the \acp{BS}\\ \hline
    $\Psi_{r}$ & The \ac{PPP} modeling the locations of the restricted areas \\ \hline
    $\Psi_{B}$ & The \ac{PHP} modeling the locations of \acp{BS} affected by restricted areas \\ \hline   
    ${\lambda_b}$; $\lambda_{r}$; $\lambda_{B}$ & The density of ${\Psi_b}$; $\Psi_{r}$; $\Psi_{B}$ \\ \hline 
    $\alpha$; $\beta$ & The value of the path-loss exponent for calculating \ac{SNR}; \ac{EMF} exposure \\ \hline
           $p_{\rm max}$ & The maximum transmit power in the uplink\\ \hline
           ${j}$ & ${j} =\sqrt{-1}$ is the imaginary unit\\ \hline
           $\rm {SAR}^{UL}$; $\rm {SAR}^{DL}$ & The value of \ac{SAR} in the uplink; downlink \\ \hline
    \end{tabular}}
\end{center}
\label{tab:TableOfNotations}
\end{table*}

\section{System Model}\label{sec:system}
This section describes the considered system of a cellular network with the \ac{EMF}-deployment constraint and corresponding stochastic geometry-based model.  
%

\begin{figure}%
 \subfigure[]{\includegraphics[width=0.42\columnwidth]{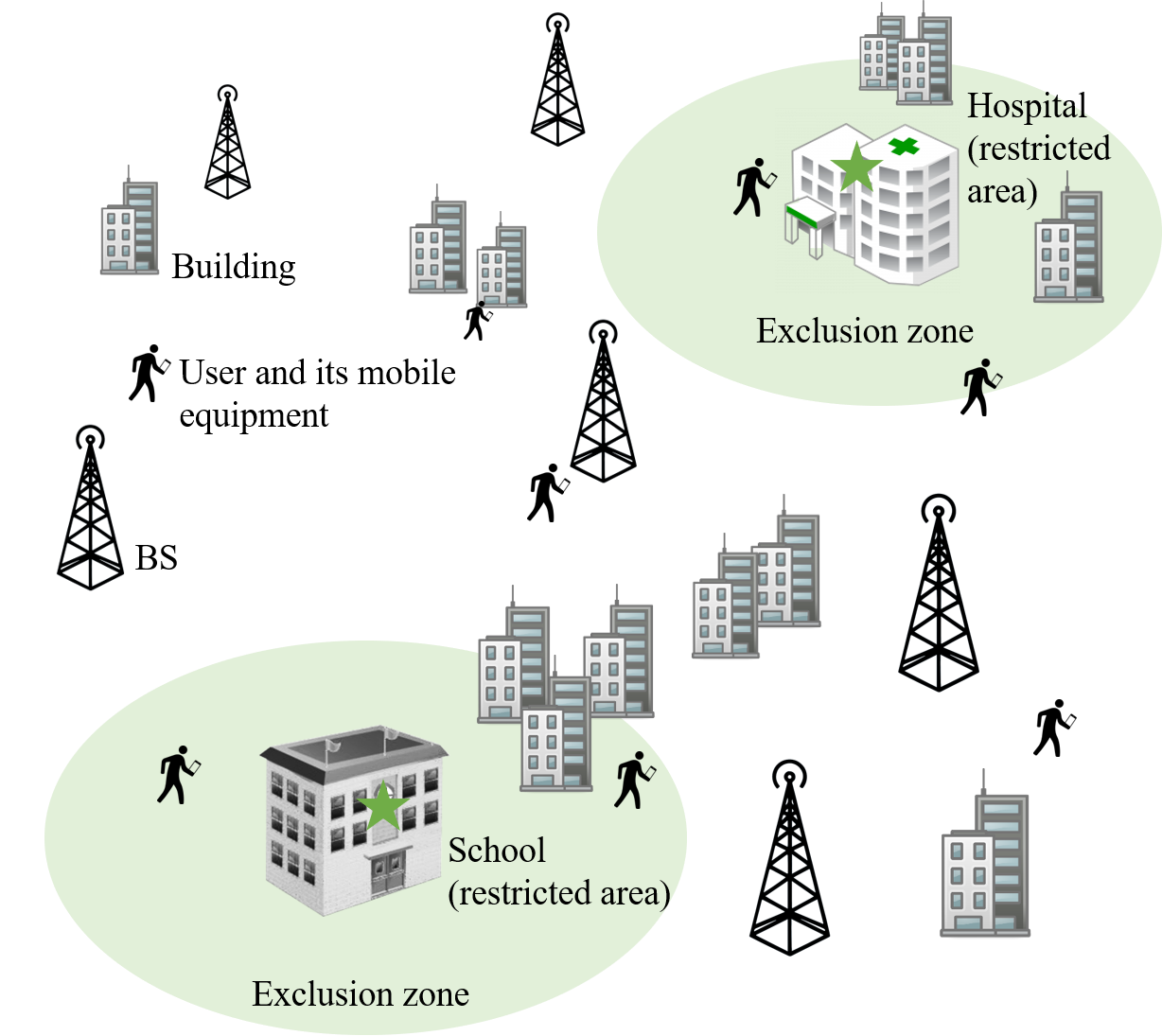}\label{fig:system1}}
 \hfill
 \subfigure[]{\includegraphics[width=0.53\columnwidth]{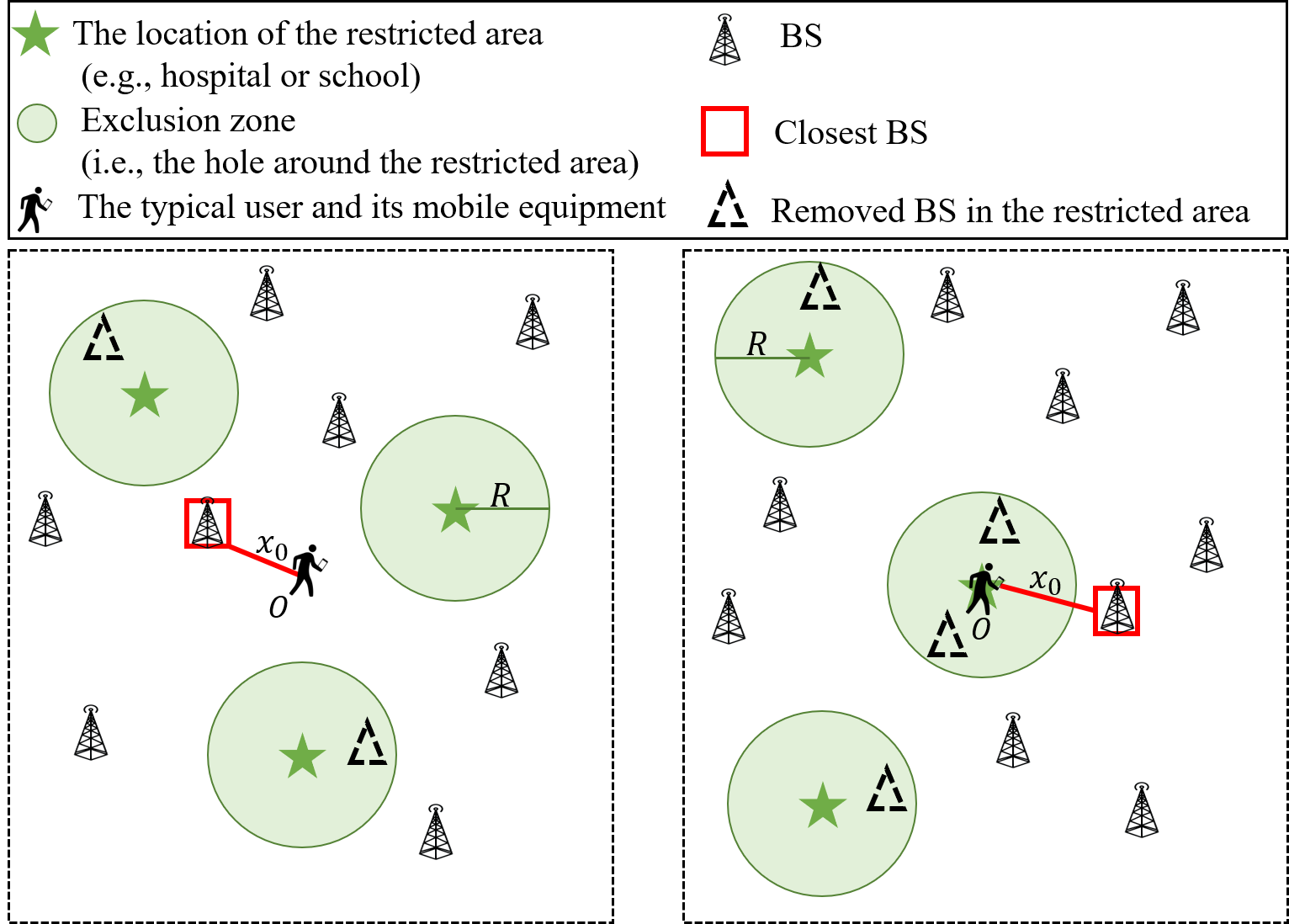}\label{fig:system2}}
 \caption{A cellular network considering the restricted areas. (a) a general scenario  (b) geometric representation (left: the typical user outside the exclusion zone, right: the typical user inside the exclusion zone).}
 \label{fig:system} 
\end{figure}

We investigate a network composed of \acp{BS} and \acp{UE}, which are equipped with omnidirectional antennas. 
\textcolor{black}{In fact, the case of directional antennas can be handled under the proposed framework with the scaled density of BSs {and antenna gain}\cite{directional}.} 
Particularly, the deployment of \acp{BS} complies with the EMF regulation on the minimum distance between \acp{BS} and the restricted areas (such as schools and hospitals). As shown in Fig.~\ref{fig:system1}, there are several restricted areas surrounded by exclusion zones/holes where \acp{BS} are not allowed (by \ac{EMF} regulation) to be deployed~\cite{hole2}.
\textcolor{black}{The radius of the exclusion zones/holes, denoted by $R$, can be regarded as the minimum distance between \acp{BS} and the restricted areas.
Using tools from stochastic geometry, the locations of BSs and users are generally assumed to follow two independent homogeneous \acp{PPP}~\cite{andrews2016primer}. However, considering the exclusion zones around restricted areas, we propose to model the locations of the \acp{BS} as a \ac{PHP}.} The \ac{PHP} is generated using two independent \acp{PPP}:
\begin{enumerate}
    \item the baseline \ac{PPP} 
    ${\Psi}_b\equiv\{b_i\}\subset\mathbb{R}^2$ with density ${\lambda}_{b}$.
    \item  the \ac{PPP} modeling the locations of the restricted areas, which represent the centers of the exclusion zones (i.e., holes), $\Psi_r\equiv\{r_i\}\subset\mathbb{R}^2$ with density $\lambda_{r}$.
\end{enumerate}
Therefore, the locations of the \acp{BS} construct a \ac{PHP} $\Psi_B$, formally defined as follows
\begin{align}
\Psi_B=\left\{b_i\in{\Psi}_b:b_i\notin\bigcup_{r_i\in\Psi_r}\mathcal{B}(r_i,R) \right\},
\end{align}
where $\mathcal{B}(r_i,R)$ is an exclusion zone/hole, a disk centered at $r_i$ with radius $R$. The density of $\Psi_B$ can be approximated as
\begin{equation}\label{eq:densityB}
\lambda_{B}={\lambda_{b}} \exp(-\lambda_{r}R^2).
\end{equation}
%
Fig.~\ref{fig:system2} presents a realization of \ac{PHP}-distributed \acp{BS}, marked by triangles with solid black borders. \acp{BS} that fall in the exclusion zones (i.e. green circles) are removed and become dashed triangles. 
\textcolor{black}{Besides, a \ac{BS} only serves a single \ac{UE} in each time-frequency recourse block. The nearest association rule is considered for the connection between BSs and users, i.e., a user is connected with its nearest BS.
Without loss of generality, we focus on the analysis on a typical user located at the origin of the network~\cite{haenggi2012stochastic}.}
Due to the limitation of the restricted areas, the location of the typical user (inside/outside the exclusion zone) impacts its performance. Therefore, in the following, we distinguish the analysis of the typical user in two different locations: outside or inside the exclusion zone, as shown in Fig.~\ref{fig:system2}. 



%


\section{\ac{EMF}-aware Downlink Exposure}\label{sec:downlink}
In this section, the downlink performance metrics are derived {to characterize} the statistics of the \ac{EMF} exposure, the coverage probability, and the allowable distance between a \ac{BS} and a user. Then, we formulate an optimization problem to maximize the downlink coverage probability subject to the \ac{EMF} exposure constraint.

\subsection{Performance Metrics: Downlink Exposure and Coverage}\label{sec:probelmDL}
The downlink \ac{EMF} exposure is quantified by the received power density from all \acp{BS} over the territory. 
Following the system model discussed in Sec.~\ref{sec:system}, the output power of a \ac{BS} is 
\begin{align}\label{metrics:PtranDL}
p=P_t G_t,
\end{align}
where $P_t$ is the transmit power of the \ac{BS} and $G_t$ is the antenna gain.
{Based on the Friis law, the received power at the typical user from a \ac{BS} with horizontal distance $x_i$ is given by
\begin{equation}\label{eq:Prec}
P_{\rm rec}(x_i)=\frac{p\, G_u\, \eta\, H_i}{(x_i^2+h^2)^{\beta/2}},
\end{equation}
where $G_u$ is the user's antenna gain, {$\eta=(\frac{c/f}{4 \pi d_0})^2$ represents the path loss at reference distance $d_0=1~\rm m$ (with frequency $f$ and wave speed $c=3\times10^8 ~\rm m/s$),} $h$ is the height of \acp{BS}, $\beta$ is the path-loss exponent, and
\textcolor{black}{
$H_i$ is the small-scale fading gain with an expected value of $1$, i.e., $\mathbb{E} \left \{ H_i \right \} =1$. More specifically, we adopt the widely-used Nakagami-$m$ fading model (with shaping parameters given by $m$) to characterize the small-scale fading. The \ac{PDF} of $H_i$ is given by~\cite{Nakagami}
\begin{align}\label{eq:ssfading}
f_{H_i}(\omega )=\frac{m^{m} \omega^{m-1}}{\Gamma(m) } e^{-m\omega },
\end{align}
where {$\Gamma\left ( m \right )=\int_{0}^{\infty} t^{m-1}e^{-t}\mathrm{d}t$  is the Gamma function~\cite{davis1959leonhard}}.}
The corresponding received power density can be expressed as
\begin{align}
W(x_i) & = \frac{P_{\rm rec}(x_i)}{A_e}=\frac{p H_i}{4\pi (x_i^2+h^2) ^{\beta/2}},
\end{align}
where $A_e=\frac{(c/f)^2 G_u}{4\pi}$ is the antenna effective area of the typical user. }
Since the \ac{BS} height is negligible compared with the horizontal distance $x_i$, for simplicity, the Euclidean distance between a BS and the typical user is approximated as their horizontal distance in the rest of the paper.
Hence, for the typical user located at the origin, the downlink \ac{EMF} exposure is
\begin{align}\label{metrics:WDL}
W^{\rm DL}=\sum_{i,b_i\in\Psi_B}\frac{p H_i}{4\pi x_i ^\beta}.
\end{align}
For this setup, our objective is to ensure that the value of $W^{\rm DL}$ is below the maximum allowed value by regulatory authorities with a high probability\cite{ITU5G:19,inversion}.  
In particular, the below condition (termed downlink \ac{EMF} constraint)
is required to be satisfied
\begin{align}\label{metrics:constrainDL}
\mathbb{P}(W^{\rm DL}\leq W_{\max})\geq\rho,
\end{align}
where $W_{\max}$ is the maximum power density specified by the guidelines. The value of $\rho$ can be chosen as $0.95$, according to the assessment of compliance regulation by \ac{ITU} \cite{ITU5G:19}. This is opposed to the deterministic approaches where $W^{\rm DL}$ is not considered as a \ac{r.v.}, and it should be strictly less than $W_{\max}$, i.e., $\rho \equiv 1$ \cite{Chiaraviglio2018planning}.

From \eqref{eq:Prec}, when the typical user associates with its closest \ac{BS} with a distance of $x_0$, the received power in the downlink can be expressed as
\begin{equation}\label{eq:PrecDL}
{
P^{\rm DL}_{\rm rec}(x_0)=\frac{p \eta H_0}{ x_0^{\alpha}},
}
\end{equation}
where the \ac{BS} height is ignored, $G_u$ is assumed to be $1$, and $\alpha$ is the path-loss exponent.
The instantaneous \ac{SNR} in the downlink is given by
\begin{align}\label{metrics:snrDL}
{{\rm SNR^{DL}}=\frac{P^{\rm DL}_{\rm rec}(x_0)}{\sigma^2} =\frac{p \eta H_0}{ x_0^{\alpha}\sigma^2},}
\end{align}
where $\sigma^2$ is noise power, \textcolor{black}{and $H_0$ is the small-scale fading  following \eqref{eq:ssfading}}. 
Note that the path-loss exponent is represented by two different notations, i.e., $\alpha$ and $\beta$, for calculating the \ac{EMF} exposure and SNR, respectively, which enables us to analyze the worst case and the typical case by setting different relations between $\alpha$ and $\beta$.
For a given value of ${\lambda}_B$, the downlink coverage probability can be defined as follows,
%
\begin{align}\label{metrics:covDL}
	\mathcal{P}^{\rm DL}_{\rm cov}({\lambda}_B)=\mathbb{P}({\rm SNR^{\rm DL}}>\tau),
\end{align}
where $\tau$ is a predefined threshold.

\subsection{Performance Analysis}
This subsection provides several steps to derive the mathematical expressions for the performance metrics formally defined in Sec.~\ref{sec:probelmDL} from the perspective of the typical user outside or inside the hole, respectively.
For instance, considering that \acp{BS} are excluded from the exclusion zones, the distance between the typical user located in the restricted area and its nearest \ac{BS} must be larger than the radius of the exclusion zone.
\subsubsection{\textbf{Distance Distribution to the Closest \ac{BS}}}
Let $X_{\rm out}$ and $X_{\rm in}$ denote the distance from the typical user to its closest \ac{BS} given that the user is outside or inside the hole, respectively. The following lemma presents the distribution of the contact distances $X_{\rm out}$ and $X_{\rm in}$.
\begin{lemma}\label{lemma:ClosestDL}
The \ac{PDF} of the distance between the typical user and the closest \ac{BS} is denoted by
\begin{subequations}\label{eq:fX}
\begin{align}
f_{X_{\rm out}}(x)&=2\pi\lambda_{B}x\exp\left ( -\lambda_{B}\pi x^2 \right ) ,x\ge 0,\label{eq:fXout}\\
f_{X_{\rm in}}(x)&=\left\{\begin{matrix}
 2\pi\lambda_{B}x\exp\left ( -\lambda_{B}\pi(x^2-R^2) \right ), & {\rm if}~x\ge R\\
 0\hfill, & {\rm if}~x< R,
\end{matrix}\right.\label{eq:fXin}
\end{align}
\end{subequations}
where $\lambda_{B}={\lambda_{b}} \exp(-\lambda_{r}R^2)$, (\ref{eq:fXout}) is for the user outside the hole, and (\ref{eq:fXin}) is for the user inside the hole.

\begin{proof}
See Appendix~\ref{app:lemma1}.
\end{proof}

\end{lemma}
\subsubsection{\textbf{Coverage Probability}}
The probability that \ac{SNR} is above a predefined threshold, i.e., the \ac{CCDF} of \ac{SNR}, is used as a metric to describe \ac{QoS} of the cellular network. Based on the distance distribution of $X_{\rm out}$ and $X_{\rm in}$, we develop the expression of the downlink coverage probability relative to a \ac{SNR} threshold $\tau$.
\begin{theorem}\label{theor:covDL} 
The downlink coverage probability of the typical user served by its closest \ac{BS} is 
%
\begin{equation}\label{eq:covDL}
{\mathcal{P}}^{\rm DL}_{\rm cov}=\int_{0}^{\infty} \sum_{k=0}^{m-1} \frac{(s^{\rm DL}\sigma^{2})^k}{k!}\exp\left ( -s^{\rm DL}\sigma^2\right )f_{X_{v}}(x_0)\,\mathrm{d}x_0,~s^{\rm DL}=\frac{m\tau}{p \eta x_0^{-\alpha}},~  v\in \left \{ \rm in ,out \right \},
\end{equation}
where $v$ represents whether or not the typical user is in a restricted area.
\end{theorem}
\begin{proof}
See Appendix~\ref{app:theorem1}.
\end{proof}

\subsubsection{\textbf{\ac{EMF} Constraint}} The \ac{EMF} exposure is required to be below the maximum allowed value ($W_{\max}$) by the authority regulations with a high probability ($\rho$).
(\ref{metrics:constrainDL}) can be rewritten in terms of the \ac{CDF} of the downlink power density, i.e., $F_{W^{\rm DL}}(w)$, as
\begin{align}
\begin{split}
 F_{W^{\rm DL}}(W_{\max})\geq\rho.
\end{split}
\end{align}
Therefore, the \ac{EMF} constraint can be expressed as 
\begin{equation}\label{eq:inverseF}
    F^{-1}_{W^{\rm DL}}(\rho) \leq W_{\max},
\end{equation}
where $F^{-1}_{W^{\rm DL}}(\rho)$ is the inverse function of $F_{W^{\rm DL}}(w)$.

\begin{theorem}\label{theor:emfDL}
The \ac{CDF} of the downlink \ac{EMF} exposure is given by
\begin{equation}\label{eq:FLemfDL}
\begin{split}
F_{W^{\rm DL}}(w)=
\end{split}\frac{1}{2} -\frac{1}{2 {j} \pi}\int_{0}^{\infty}\frac{1}{t}\left [ e^{-{j} tw}\mathcal{L}_{W^{\rm DL}}(-{j} t)- e^{{j} tw}\mathcal{L}_{W^{\rm DL}}({j} t) \right ] {\rm d} t,   
\end{equation}
where $\mathcal{L} _{W^{\rm DL}}(s)$ is the Laplace transform of downlink {\ac{EMF}} exposure and 
\begin{equation}\label{eq:LemfDL}
\begin{split}
\mathcal{L} _{W^{\rm DL}}(s)
&=\textcolor{black}{\exp\left ( -2\pi\lambda_B\int_{v(R)}^{\infty}\left [ 1-\kappa^{\rm DL}(x,s)  \right ]  x\,\mathrm{d}x \right )}
,\\
\kappa^{\rm DL}(x,s)& =\left(\frac{m}{m+sp(4\pi)^{-1}x^{-\beta}}\right )^m,
\end{split}
\end{equation}
where $R$ is the radius of the hole, $\lambda_B$ is the density of \ac{PHP}-distributed \acp{BS}, and $v(R)=0$ if the typical user is outside the hole, otherwise $v(R)=R$.

\end{theorem}
\begin{proof}
See Appendix~\ref{app:theorem2}.
\end{proof}

\subsubsection{\textbf{Compliance Distance Between Users and \acp{BS}}} \label{subsec:Xcom}
The distance between a BS and a user should comply with the \ac{EMF} constraint, i.e., conditioned on the serving \ac{BS} located at a compliance distance $x_c$ to the typical user, the corresponding conditional \ac{EMF} exposure should not exceed the maximum allowable limit ($W_{\max}$) with a high probability ($\rho$). Namely, $x_c$ satisfies the conditional \ac{EMF} constraint as follows,
\begin{equation}\label{eq:inverseF_x0}
    F_{W^{\rm DL}|X_{\rm out}=x_c}(W_{\max}) \geq \rho,
\end{equation}
where $F_{W^{\rm DL}|X_{\rm out}=x_c}(W_{\max})$ is the \ac{CDF} of the downlink \ac{EMF} exposure conditioned on the distance between the serving \ac{BS} and the user being $x_c$.
The minimum value of $x_c$ that satisfies (\ref{eq:inverseF_x0}) is denoted by $x_{\rm com}$ and is given by
\begin{align}
    x_{\rm com} \triangleq \inf_{x\in \mathbb{R}} \left\{x:   F_{W^{\rm DL}|X_{\rm out}=x}(W_{\max}) \geq \rho  \right\}.
\end{align}
No public access is allowed to the area centered at a \ac{BS} with the radius of $x_{\rm com}$, since people inside will experience downlink \ac{EMF} exposure above the safety threshold. 
%

\begin{theorem}\label{theor:emfx0}
The \ac{CDF} of the {\ac{EMF}} exposure conditioned on the serving \ac{BS} being at a distance $x_0$ from the typical user is given by
\begin{equation}\label{eq:FLemfDLx0}
\begin{split}
F_{W^{\rm DL}|x_0}(w)
&=\frac{1}{2} -\frac{1}{2{j} \pi}\int_{0}^{\infty}\frac{1}{t}\left [ e^{-{j} tw}\mathcal{L}_{W^{\rm DL}|x_0}(-{j} t)- e^{{j} tw}\mathcal{L}_{W^{\rm DL}|x_0}({j} t) \right ]\mathrm{d}t ,  
\end{split}
\end{equation}
where $\mathcal{L}_{W^{\rm DL}|x_0}(s)$ is the Laplace transform of conditional \ac{EMF} exposure and 
\begin{equation}\label{eq:X_com}
\mathcal{L}_{W^{\rm DL}|x_0}(s)
{=}
\textcolor{black}{\kappa^{\rm DL}(x_0,s)\exp\left ( -2\pi\lambda_B\int_{x_0}^{\infty}\left [ 1-\kappa^{\rm DL}(x,s)  \right ]  x\,\mathrm{d}x \right )},
\end{equation}
where $\lambda_B$ is given in (\ref{eq:densityB}) and $\kappa^{\rm DL}(\cdot)$ is given in \eqref{eq:LemfDL}.
\end{theorem}
\begin{proof}
See Appendix~\ref{app:theorem3}.
\end{proof}

When setting $w=W_{\max}$ and $\rho$ as a constant defined by the \ac{ITU}, $F_{W^{\rm DL}|x_0}(w)$ is a function of $x_0$. Using its inverse function, we can find the minimum compliance distance between a user and a \ac{BS}, $x_{\rm com}$. Thus, the compliance distance is $x_c \ge x_{\rm com}$.

\subsection{Optimal \ac{EMF}-aware Design}\label{subsec:optDL}
We try to maximize the downlink coverage probability while ensuring that the downlink \ac{EMF} constraint is always satisfied. This optimization problem can be formally defined as follows,
%
\begin{equation}\label{eq:optproblem1}
\begin{split}
{\rm {\mathbf{OP_1}:\,}}
&\underset{{\lambda}_B\in\mathbb{Z}_{+}}{\text{maximize} }\quad\mathcal{P}^{\rm DL}_{\rm cov}({\lambda}_B)
\\&\text{subject to:} \quad\mathbb{P}(W^{\rm DL}\leq W_{\rm max})\geq\rho.
\end{split}
\end{equation}
Intuitively, the network would have better coverage performance as the \ac{BS} density $\lambda_B$ increases since the distance between the typical user and its closest \ac{BS} becomes closer and the path loss is correspondingly reduced, thereby improving the downlink \ac{SNR}.
However, the \ac{EMF} exposure is expected {to increase}, which means the optimal density $\lambda_B^\ast$ (corresponding to a specific  baseline density of \acp{BS}, $\lambda_b^\ast$) can be found by gradually increasing the value of $\lambda_B$ until $F^{-1}_{W^{\rm DL}}(\rho,\lambda_B^\ast)$ reaches $W_{\max}$. Based on (\ref{eq:densityB}), increasing $\lambda_B$ can be realized by reducing the density and the radius of holes or increasing the baseline density.

\section{\ac{EMF}-aware Uplink Exposure}\label{sec:uplink}
In this section, we focus on the \ac{EMF} exposure induced by individual mobile equipment. We define the uplink performance metrics and then provide corresponding analysis.

\subsection{Performance Metrics: Uplink Exposure and Coverage}\label{sec:probelmUL}
We now analyze the uplink exposure of the cellular network described in Sec.~\ref{sec:system}. 
The mobile equipment deploys a power control mechanism on its transmit power to compensate for the path loss, fully or partially depending on  the power control factor $\epsilon$, i.e., $\epsilon=1$ and $0<\epsilon<1$, respectively. 
The transmit power at the mobile equipment can be defined as 
\begin{equation}\label{eq:PtranUL}
P^{\rm UL}_{\rm tran}(x_0)=\left\{\begin{matrix}
{p_{u}x_0^{\alpha \epsilon}}, & x_0< X_{\rm max}\\
 p_{\rm max}, & {\rm otherwise},~
\end{matrix}\right.
\epsilon \in (0,1],
\end{equation}
where $p_u$ is a constant, $\alpha$ is the path-loss exponent, $x_0$ is the distance between the \ac{BS} and its serving user, $p_{\rm max}$ is the maximum transmit power of the mobile equipment, and $X_{\max}=\left(p_{\rm max}/p_u\right)^{1/\alpha\epsilon}$.
Let $u_0$ denote the distance from the user to its personal mobile equipment and assume that $u_0$ is in the far-filed of the transmit antenna of mobile equipment.\footnote{\textcolor{black}{Nowadays, the smartphone is not used as a plain old telephone attached to the ear, but rather it is put in front of the chest and used for, e.g., chatting, exploring social media, watching videos, or listening to streaming audio.}}
The distance from the \ac{BS} to its associated user or mobile equipment is assumed to be equivalent \textcolor{black}{since the distance from the typical user to its own mobile device is much shorter than that to its serving \ac{BS}, i.e., $u_0<<x_0$.}
Unlike the downlink exposure that comes from all \acp{BS} in the network, the uplink exposure is dominated by the user's personal mobile equipment~\cite{kuehn2019modelling,R2}. \textcolor{black}{Specifically, the transmit power at mobile equipment is quite low and attenuates significantly after long-distance transmission. Therefore, the exposure from other users' mobile devices to the typical user is negligible. However, due to the close distance between the typical user and its own mobile device, the corresponding exposure can not be ignored.} Thus,
the uplink \ac{EMF} exposure at the typical user can be assessed by the received power density from its mobile equipment (termed as typical \ac{UE}) as 
\begin{equation}\label{metrics:WUL}
W^{\rm UL}(x_0)=\frac{P^{\rm UL}_{\rm tran}(x_0)\Omega}{4\pi u_0^{\beta}}=
\left\{\begin{matrix}
 \frac{p_{u}x_0^{\alpha \epsilon }\Omega}{4\pi u_0^{\beta}}, & x_0< X_{\rm max}\\
 \frac{p_{\rm max}\Omega}{4\pi u_0^{\beta}}\hfill, & {\rm otherwise},
\end{matrix}\right.
\end{equation}
where $x_0$ is the distance between the typical user and its serving \ac{BS} (termed as tagged \ac{BS}) and \textcolor{black}{$\Omega$ follows the Nakagami-$m$ fading in \eqref{eq:ssfading} with mean $1$}. The uplink \ac{EMF} exposure is also required to be below the maximum permitted value by regulatory authorities with a high probability of $\rho$, i.e., the uplink \ac{EMF} constraint can be expressed as
\begin{align}\label{metrics:constrainUL}
\mathbb{P}(W^{\rm UL}\leq W_{\max})\geq\rho.
\end{align}
The received power at the tagged \ac{BS} from the typical \ac{UE} with a distance $x_0$ is given by
\begin{equation}\label{eq:PrecUL}
P^{\rm UL}_{\rm rec}(x_0)=\frac{P^{\rm UL}_{\rm tran}(x_0)\eta \Omega_0}{x_0^{\alpha}}=
\left\{\begin{matrix}
 {p_{u}x_0^{\alpha( \epsilon -1)} \eta  \Omega_0}, & x_0< X_{\rm max}\\
 \frac{p_{\rm max}\eta \Omega_0}{x_0^{\alpha}}\hfill, & {\rm otherwise},
\end{matrix}\right.
\end{equation}
where \textcolor{black}{$\Omega_0$ is the small-scale fading and has the same distribution as $\Omega$}, and $\eta $ is defined in \eqref{eq:Prec}.
In particular, if $\epsilon=1$, it is a full power control case, and the average received power will be a constant when $x_0<X_{\max}$.
 
The corresponding instantaneous \ac{SNR} in the uplink can be expressed as 
\begin{equation}\label{eq:SNRUL}
{\rm SNR^{UL}}=\frac{P^{\rm UL}_{\rm rec}(x_0)}{\sigma^2},
\end{equation}
which can be used to compute the uplink coverage probability as follows 
\begin{align}\label{metrics:covUL}
	\mathcal{P}^{\rm UL}_{\rm cov}({\lambda}_B)=\mathbb{P}({\rm SNR^{\rm UL}}>\tau),
\end{align}
where $\lambda_B$ is givn in (\ref{eq:densityB}) and $\tau$ is the predefined threshold.

\subsection{Performance Analysis}
Similar to the downlink case, we analyze the uplink exposure of the typical user inside and outside the hole, respectively, and the uplink coverage probability. The distribution of the distance between the typical user and its closest \ac{BS} in Lemma~\ref{lemma:ClosestDL} is still applicable to the following analysis. 
\subsubsection{\textbf{Coverage Probability}}
The following theorem gives the expression of the coverage probability under the power control mechanism, which is defined as the \ac{CCDF} of the uplink \ac{SNR}.
\begin{theorem}\label{theor:covUL} 
The uplink coverage probability of the users inside or outside the hole is given by 
%
\begin{equation}\label{eq:covUL}
\begin{split}
\mathcal{P}^{{\rm UL}}_{\rm cov}
=&\int_{0}^{X_{\rm max}} \sum_{k=0}^{m-1} \frac{(s^{\rm UL}_1\sigma^{2})^k}{k!}\exp\left ( - s^{\rm UL}_1\sigma^{2} \right ) 
f_{X_{v}}(x_0) \mathrm{d}x_0 ~+\\&
\int_{X_{\rm max}}^{\infty}  \sum_{k=0}^{m-1} \frac{(s^{\rm UL}_2\sigma^{2})^k}{k!}\exp\left ( - s^{\rm UL}_2\sigma^{2}\right )
f_{X_{v}}(x_0) \mathrm{d}x_0, 
\end{split}
\end{equation}
where $v\in\left \{ \rm in, out \right \}$, $s^{\rm UL}_1=\frac{m\tau}{ p_{u}x_0^{\alpha (\epsilon-1)} \eta  }$, $s^{\rm UL}_2=\frac{m\tau}{ p_{\rm max}\eta x_0^{-\alpha}  }$,
and $f_{X_v}(x_0)$ is given in (\ref{eq:fX}). 
\end{theorem}

\begin{proof}
\textcolor{black}{Similar to the method in Appendix~\ref{app:theorem2}.}
\end{proof}
\subsubsection{\textbf{\ac{EMF} Constraint}}
The uplink \ac{EMF} exposure at the typical user is mostly from its own mobile equipment, which should be lower than the maximum allowed value by authority regulations with a high probability. Similar to the downlink \ac{EMF} constraint, (\ref{metrics:constrainUL}) can be derived from the Laplace transform of the uplink received power density, denoted by $\mathcal{L} _{W^{\rm UL}}$ as
%
\begin{equation}\label{eq:LemfUL}
\textcolor{black}{\mathcal{L} _{W^{\rm UL}}(s)=
\int_{0}^{X_{\rm max}} \kappa^{\rm UL}_1(x_0,s){f_{X_v}(x_0)} \,\mathrm{d}x_0
+\int_{X_{\rm max}}^{\infty} \kappa^{\rm UL}_2(s){f_{X_v}(x_0)}\,\mathrm{d}x_0,}
\end{equation}
where $u_0$ is the distance between the typical user and its mobile equipment, $f_{X_{v}}(x_0)$ is defined in Lemma~\ref{lemma:ClosestDL},  $v\in\left\{\rm in, out\right\}$,  \textcolor{black}{$\kappa^{\rm UL}_1(x_0,s)=\left(\frac{m}{m+sp_{u}x_0^{\alpha \epsilon}(4\pi)^{-1}u_0^{-\beta}}\right )^m$, and $\kappa^{\rm UL}_2(s)=\left(\frac{m}{m+sp_{\rm max}(4\pi)^{-1}u_0^{-\beta}}\right )^m$.}

%
\begin{theorem}\label{theor:emfUL} 
The \ac{CDF} of the uplink {\ac{EMF}} exposure is given by
\begin{equation}\label{eq:FLemfUL}
\begin{split}
F_{W^{\rm UL}}(w)=
\end{split}\frac{1}{2} -\frac{1}{2{j} \pi}\int_{0}^{\infty}\frac{1}{t}\left [ e^{-{j} tw}\mathcal{L}_{W^{\rm UL}}(-{j} t)- e^{{j} tw}\mathcal{L}_{W^{\rm UL}}({j} t) \right ] \mathrm{d}t.  
\end{equation}
\end{theorem}
Therefore, the uplink \ac{EMF} constraint can be described by the inverse function of $F_{W^{\rm UL}}(w)$, i.e., $F^{-1}_{W^{\rm UL}}(\rho) \le W_{\max}$.

From above analysis, we notice that, under the power control mechanism, the closer distance between the user and its serving \ac{BS} leads to lower transmit power at its \ac{UE}.
Unlike what we discussed in Sec.~\ref{subsec:optDL}, increasing the number of \acp{BS} is beneficial for both the reduction on the exposure and the improvement on the coverage probability in the uplink.
Therefore, the uplink optimization problem is omitted.



\section{Joint Downlink\&Uplink Exposure}\label{sec:DLandUL}
In most situations, a person is exposed to \ac{EMF} from both his/her own mobile equipment and \acp{BS}.
In this section, we consider a metric called exposure index that accounts for both the uplink and  downlink exposure in the cellular network discussed in Sec.~\ref{sec:system}, followed by the performance analysis and the optimal \ac{EMF}-aware network design.
\subsection{Performance Metrics: Exposure Index}\label{sec:probelmULandDL}
In~\cite{sar}, a metric that quantifies the population exposure to \ac{EMF} is introduced. This metric considers both the uplink and downlink \ac{EMF} exposure. It also accounts for user-specific properties such as age (adult or child), usage (data or voice call), and posture (standing or sitting).
{
The total exposure (termed exposure index) at the typical user can be expressed as the sum of uplink exposure index (${\rm EI^{UL}}$) and downlink exposure index (${\rm EI^{DL}}$), which is given by~\cite{sar}
%
\begin{equation}\label{eq:EI}
\begin{split}
{\rm EI}(x_0)&={\rm EI^{UL}}(x_0)+{\rm EI^{DL}}={\rm SAR^{UL}}P^{\rm UL}_{\rm tran}(x_0)+{\rm SAR^{DL}}W^{\rm DL} 
\\&={\rm SAR^{UL}}P^{\rm UL}_{\rm tran}(x_0)+{\rm SAR^{DL}}\sum_{i,b_i\in \Psi_B}\frac{pH_i}{4\pi x_i^{\beta}} \quad \left [ \rm W/kg \right ],
\end{split}
\end{equation}
%
where $x_0$ is the distance between the typical user and the serving \ac{BS} located ar $b_0$,
$P^{\rm UL}_{\rm tran}$ is the uplink transmit power defined in (\ref{eq:PtranUL})}, $W^{\rm DL}$ is the downlink received power density defined in (\ref{metrics:WDL}), ${\rm SAR^{UL}}$ is the reference induced \ac{SAR} in the uplink when the transmit power from the \ac{UE} is unity, and ${\rm SAR^{DL}}$ is the reference induced \ac{SAR} in the downlink when the received power density from the \ac{BS} at the \ac{UE} is unity. More clearly,  ${\rm SAR^{UL}}\, \left(\rm \frac{W}{kg}/W\right)$ is normalized to unit transmit power and ${\rm SAR^{DL}}\,  \left(\rm \frac{W}{kg}/\frac{W}{m^2}\right)$ is normalized to unit power density. It is worth noting that the reference \ac{SAR} depends on the user-specific properties mentioned earlier.


\subsection{Performance Analysis}

The exposure index in (\ref{eq:EI}) can be divided into two parts, i.e., one is related to the serving \ac{BS} at $b_0$ with distance $x_0$ and the other is from the rest of \ac{BS}, as follows.
\begin{equation}\label{eq:EI1}
\begin{split}
{\rm EI}(x_0)=\left ( {\rm SAR^{UL}}P^{\rm UL}_{\rm tran}(x_0)+ {\rm SAR^{DL}}\frac{pH_0}{4\pi x_0^{\beta}}\right ) 
+ {\rm SAR^{DL}}\sum_{i,b_i\in \Psi_B \setminus \{b_0\}}\frac{pH_i}{4\pi x_i^{\beta}}.
\end{split}
\end{equation}
%
The Laplace transform of the exposure index ${\rm EI}(x_0)$ conditioned on $x_0$ is given by
%
\begin{equation}\label{eq:LEIx0}
\begin{split}
\textcolor{black}{\mathcal{L}_{{\rm EI}|x_0}(s)={\exp(-s{\rm SAR^{UL}}P^{\rm UL}_{\rm tran}(x_0))}\kappa^{\rm J}(x_0,s)
\exp\left ( -2\pi\lambda_B\int_{x_0}^{\infty}[1-\kappa^{\rm J}(x,s)]\,x\,\mathrm{d}x \right ),}
\end{split}
\end{equation}
where $\kappa^{\rm J}(x_0,s)=\left(\frac{m}{m+s{\rm SAR^{DL}}p(4\pi)^{-1} x_0^{-\beta}}\right)^m$.
\begin{proof}
The Laplace transform of ${\rm EI}(x_0)$ in \eqref{eq:EI1} is given by
\begin{equation}\label{eq:EI2}
\begin{split}
&\mathcal{L}_{{\rm EI}|x_0}(s)=\mathbb{E_{\rm EI}}\left [ \exp(-s{\rm EI}(x_0)) \right ]
 \\&=\mathbb{E_{\rm EI}}\left [ \exp\left [ -s\left ( \!{\rm SAR^{UL}}P^{\rm UL}_{\rm tran}(x_0)+ \!{\rm SAR^{DL}}\frac{p H_0}{4\pi x_0^{\beta}}\right ) 
\!-\!s \left ( {\rm SAR^{DL}}\!\sum_{i,b_i\in \Psi_B \setminus \{b_0\}}\frac{pH_i}{4\pi x_i^{\beta}} \right ) \right ]   \right ]
 \\&=\exp(-s{\rm SAR^{UL}}P^{\rm UL}_{\rm tran}(x_0))
\mathbb{E}_{H_0}\left [ \exp\left ( -s{\rm SAR^{DL}}\frac{pH_0}{4\pi x_0^{\beta}} \right )  \right ]\times
\\&\quad\,\, \mathbb{E}_{\Psi_B,\left \{ H_i \right \}}\left [ \prod_{i,b_i\in \Psi_B \setminus \{b_0\}} \exp\left ( -s{\rm SAR^{DL}}\frac{pH_i}{4\pi x_i^{\beta}} \right )   \right ] .
\end{split}
\end{equation}
Then, we process (\ref{eq:EI2}) by following the similar approaches in Appendix~\ref{app:theorem3} and thus obtain \eqref{eq:LEIx0}.
\end{proof}

Based on the distance distribution of $x_0$ in Lemma~\ref{lemma:ClosestDL}, the unconditional Laplace transform of the exposure index can be obtained as 
\begin{equation}\label{eq:LEI}
\begin{split}\
\mathcal{L}_{{\rm EI}}(s)=\int_{0}^{\infty} \mathcal{L}_{{\rm EI}|x_0}(s)f_{X_{v}}(x_0)\,\mathrm{d} x_0,
\end{split}
\end{equation}
where $f_{X_{v}}(x_0)$ is given in (\ref{eq:fX}).
Additionally, the Laplace transform of ${\rm EI^{UL}}$ and ${\rm EI^{DL}}$ can be expressed as
\begin{subequations}
\begin{align}
\mathcal{L}_{{\rm EI^{UL}}}(s)&\!=\!\mathbb{E}\left [ \exp\left ( {\rm SAR^{UL}}P^{\rm UL}_{\rm tran}(x_0) \right )  \right ]  
\!=\!\int_{0}^{\infty} \exp\left ( -s{\rm SAR^{UL}}P^{\rm UL}_{\rm tran}(x_0)  \right ) f_{X_{v}}(x_0)\,\mathrm{d} x_0,
\label{eq:EIul}
\\
\mathcal{L}_{{\rm EI^{DL}}}(s)&\!=\!\mathbb{E}\left [ \exp\left ( {\rm SAR^{DL}}W^{\rm DL} \right )  \right ] 
\!=\! \textcolor{black}{\exp\left (-2\pi\lambda_B\int_{v(R)}^{\infty}\left [ 1-\kappa^J(x,s)  \right ] x\,\mathrm{d} x  \right ) ,}
\label{eq:EIdl}
\end{align}
\end{subequations}
where \eqref{eq:EIul} is from $\mathbb{E}[g(t)]=\int_{0}^{\infty} g(t)f_T(t)\mathrm{d}t$, $f_T(t)$ is the \ac{PDF} of $t$, and \eqref{eq:EIdl} can be obtained by following the similar methods to Appendix~\ref{app:theorem2}. Based on the Gil-Pelaez theorem in (\ref{eq:CDF}), we derive the \ac{CDF} of $\rm EI$ in the following theorem.
\begin{theorem}\label{theor:EI}
The \ac{CDF} of the total exposure including the uplink exposure and the downlink exposure (i.e. exposure index) is given by
\begin{equation}\label{eq:FLEI}
\begin{split}
F_{\rm EI}(\varepsilon)=
\end{split}\frac{1}{2} -\frac{1}{2{j} \pi}\int_{0}^{\infty}\frac{1}{t}\left [ e^{-{j} t\varepsilon}\mathcal{L}_{{\rm EI}}(-{j} t)- e^{{j} t\varepsilon}\mathcal{L}_{{\rm EI}}({j} t) \right ] \mathrm{d} t.  
\end{equation}


\end{theorem}


\subsection{Optimal \ac{EMF}-aware Design}\label{subsec:optDLUL}
As we mentioned before, increasing the density of \acp{BS} is effective to enhance the coverage performance both in the downlink and uplink and reduce the uplink \ac{EMF} exposure. On the contrary, increasing the density of \acp{BS} has a negative influence on the downlink \ac{EMF} exposure. Therefore, the proper design of ${\lambda}_B$ is essential to minimize both the uplink and downlink exposure. In the following, we formulate an optimization problem to find the optimal density of \acp{BS} that  minimizes $\rm EI$ in \eqref{eq:EI1} (that quantifies the joint downlink\&uplink exposure). Since $\rm EI$ is a random variable, we optimize its $\rho$-th percentile as follows  
\begin{equation}\label{eq:optproblem3}
\begin{split}
{\rm {\mathbf{OP_3}:\,}}
&\underset{{\lambda}_B\in\mathbb{Z}_{+}}{\text{minimize} }\quad F^{-1}_{\rm EI}({\rho, \lambda}_B),
\end{split}
\end{equation}
where $F^{-1}_{\rm EI}$ is the inverse \ac{CDF} of ${\rm EI}$. 
The optimal ${\lambda}_B^\ast$ can be found by any efficient one dimension search algorithm 
{such as bi-sectional~\cite{eiger1984bisection} or golden section methods~\cite{kiefer1953sequential}.}
Based on (\ref{eq:densityB}), ${\lambda}_B^\ast$ is corresponding to a specific baseline density of \acp{BS}, $\lambda_b^\ast$, and a specific radius of the exclusion zone, $R^\ast$.


\section{Numerical Results and Discussion}\label{sec:simu}

In this section, Monte Carlo simulations and numerical results for the analytical expressions are conducted. 
The performance of the \ac{EMF}-aware cellular network is investigated for three scenarios, i.e., the downlink, uplink, and  joint downlink\&uplink.  The parameters used in numerical results and their default values are given in Table.~\ref{tab:simulation}. Some of them are swept to explore their effects on the \ac{EMF} exposure and coverage probability.
\textcolor{black}{Without loss of generality, we consider the Nakagami-$m$ fading with $m=1$.}

\begin{table}[t!]\caption{Table of System Numerical Parameters.}
\centering
\begin{center}
\resizebox{\textwidth}{!}
{
\renewcommand{\arraystretch}{1.2}\small
    \begin{tabular}{ | {c} | {c} || {c} | {c} | }
    \hline
        \hline
    \textbf{System Parameters} & \textbf{Default Values} &\textbf{System Parameters} & \textbf{Default Values}  \\ \hline
    $\lambda_r$; $\lambda_b$  & $10^{-6} ~{\rm{holes/m^2}}$; $10^{-5} ~\rm{BSs/m^2}$ & $\alpha$; $\beta$ & $4$; $2.5$\\ \hline
    $f$ & $2600~\rm MHz$ & $\eta$ & $-40~\rm dB$ \\ \hline
    $G_t$ & $15 ~\rm dB$ & $P_t$ & $200 ~\rm W$ \\ \hline
    $p_{u}$& $0.008~{\rm mW}$& 
    $p_{\max}$ & $200~{\rm mW}$ \\ \hline
    $\epsilon$ & $0.4$ & {$m$} & {$1$}
    \\ \hline
    $u_0$& $20 ~\rm cm$ & $R$ & $50 ~\rm m$\\ \hline
          $\sigma^2$(downlink; uplink) & {$10^{-11}~ \rm W$; $10^{-12} ~\rm W$} & $\rm SAR^{UL}$; $\rm SAR^{DL}$ & $0.0053~\rm \frac{W}{kg}/W$; $0.0042~\rm \frac{W}{kg}/\frac{W}{m^2}$\\ \hline
    \end{tabular}}
\end{center}
\label{tab:simulation}
\end{table}

\subsection{Worst-Case Scenario for Downlink Exposure}\label{sec:simuDL}
The downlink simulations are firstly conducted under the conservative setting, where we consider the maximum transmit power ($200~\rm W$) and 
{the maximum antenna gain ($15~\rm dB$) at 5G \acp{BS}}~\cite[Table 6]{chiaraviglio2021health}. 
{Based on the recommendation provided in 3GPP TR 36.814–900~\cite{access2010further}, as a conservative setting, the path-loss exponent is set as $\alpha=4$ while analyzing the downlink coverage probability (by considering that the communication link is blocked) and as $\beta=2.5$ while computing the downlink \ac{EMF} exposure (for the environment free of blockages).}

\begin{figure}[t!]
\begin{minipage}{.49\textwidth}
    \centering
    \includegraphics[width=1\linewidth]{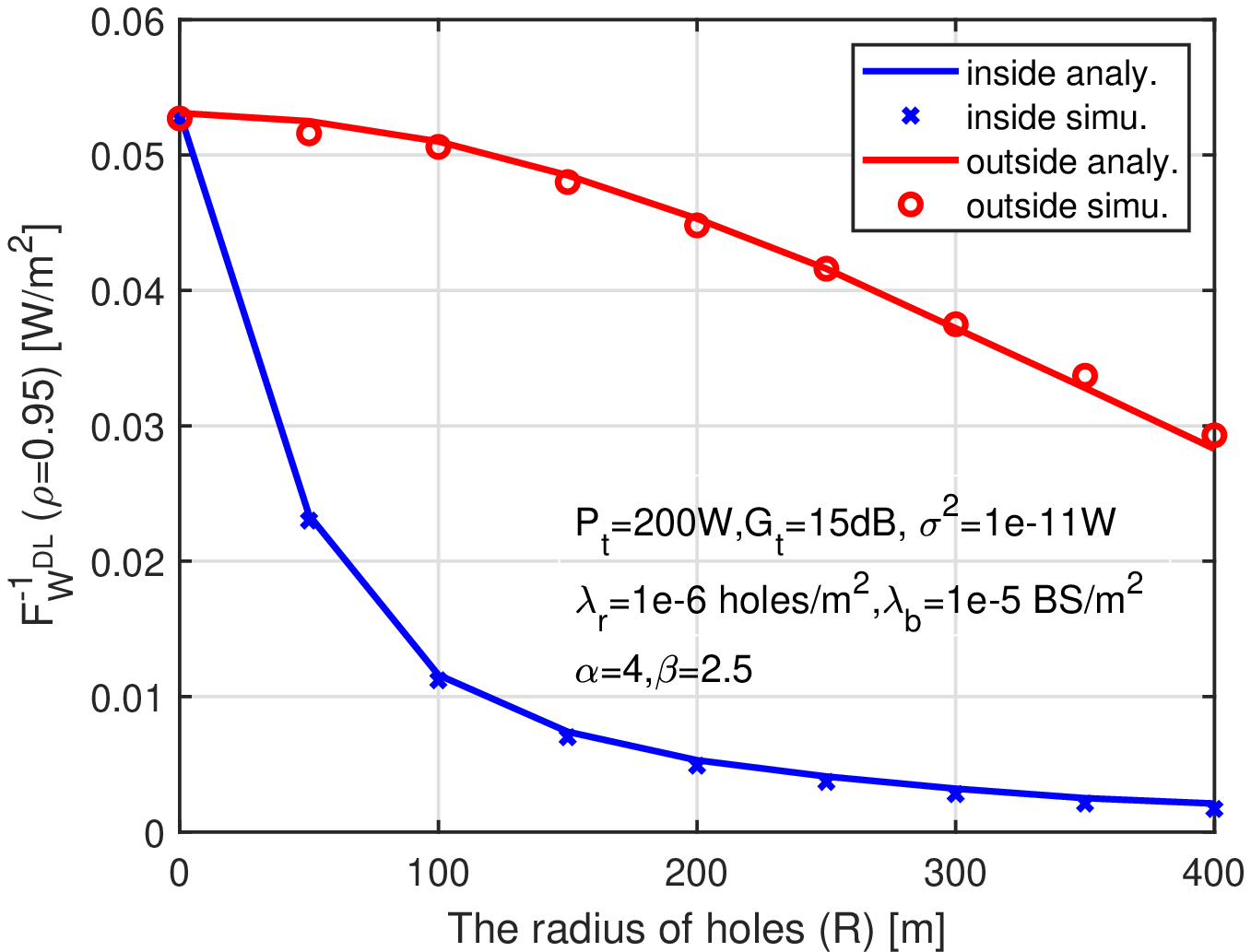}
    \caption{The $95$-th percentile of exposure level in the downlink (i.e., power density induced at the typical user) inside or outside the restricted area for various exclusion zone radii.}
\label{fig:Rholes_EMF_DL}
\end{minipage}
\hfill
\begin{minipage}{.49\textwidth}
    \centering
    \includegraphics[width=1\linewidth]{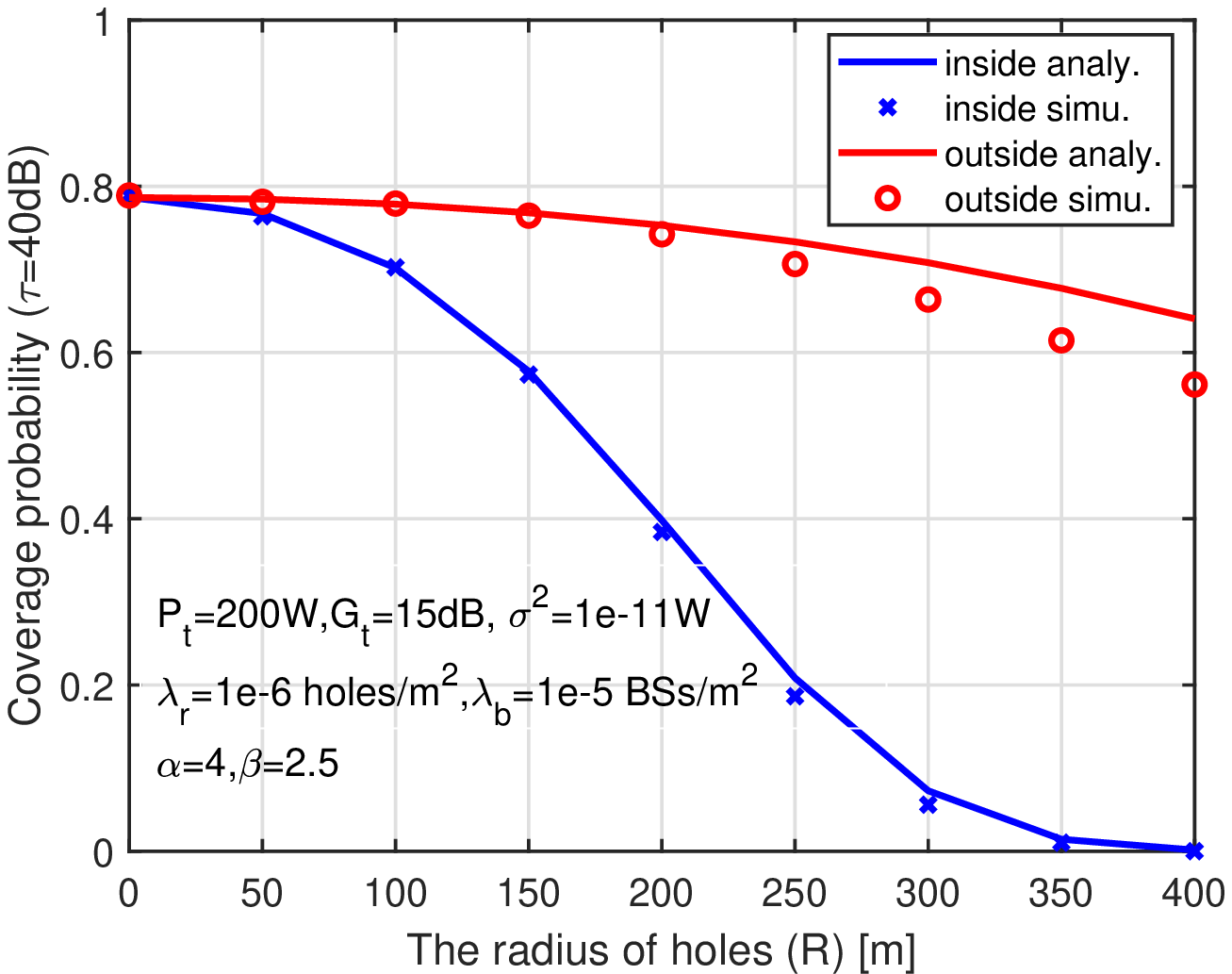}
    \caption{Focusing on $\tau=40\rm dB$, the \ac{CCDF} of the downlink \ac{SNR} levels inside or outside the restricted area for various exclusion zone radii.}
    \label{fig:Rholes_SNR_DL}
    \end{minipage}
\end{figure}
In Fig.~\ref{fig:Rholes_EMF_DL}, we compare the statistical exposure level inside and outside the restricted area for different minimum allowable distance between the \ac{BS} and restricted areas, $R$. In particular, the statistical level, denoted by $F^{-1}_{W^{\rm DL}}(\rho)$, means that the exposure will not exceed this level with a high probability of $\rho$, as explained in (\ref{eq:inverseF}). According to the compliance regulations~\cite{ITU5G:19}, we set $\rho=0.95$, and $F^{-1}_{W^{\rm DL}}(0.95)$ represents the $95$-th percentile of exposure level in the downlink. 
The \ac{EMF} exposure level without considering the minimum distance between \acp{BS} and restricted areas (i.e., the hole radius $R=0~\rm m$) has been studied in~\cite{ICNIRPGuidelines:18,SaudiEMF:21}.
It can be observed from Fig.~\ref{fig:Rholes_EMF_DL} that the regulation on the minimum distance between \acp{BS} and restricted areas greatly affects the measurement of EMF exposure. 
{Obviously, the exposure of the typical user outside the hole is higher than the one inside the restricted area, as the distance between the BS and the typical user in the restricted area is always larger than the hole radius $R$.}
Moreover, Fig.~\ref{fig:Rholes_EMF_DL} illustrates that a rise in the radius of the hole dramatically reduces the exposure for the user inside the hole and still has a slight impact on the user outside the hole. This is because a larger value of $R$ not only leads to the longer propagation distance but slightly decreases the density of \acp{BS} in (\ref{eq:densityB}), resulting in fewer \acp{BS}, namely, fewer radiating sources.

In Fig.~\ref{fig:Rholes_SNR_DL}, the coverage probability is depicted against the hole radius, for the same settings considered in Fig.~\ref{fig:Rholes_EMF_DL}.  
We notice that the typical user outside the restricted areas experiences better  coverage than the user inside. We can also observe the slight reduction in the coverage of the users outside the restricted areas as we increase the hole radius (for similar reasons as explained in our comments on Fig.~\ref{fig:Rholes_EMF_DL}).
By comparing Fig.~\ref{fig:Rholes_EMF_DL} and Fig.~\ref{fig:Rholes_SNR_DL}, it is noticed that the hole with a large radius can protect people from serious \ac{EMF} exposure while impairing communication quality. It is also worth noting that at larger values of $R$, the gap between simulation and analytical results is caused by the \ac{PHP} approximation of the \ac{BS} distribution.

\begin{figure}[t!]
\begin{minipage}{.49\textwidth}
    \centering
    \includegraphics[width=1\linewidth]{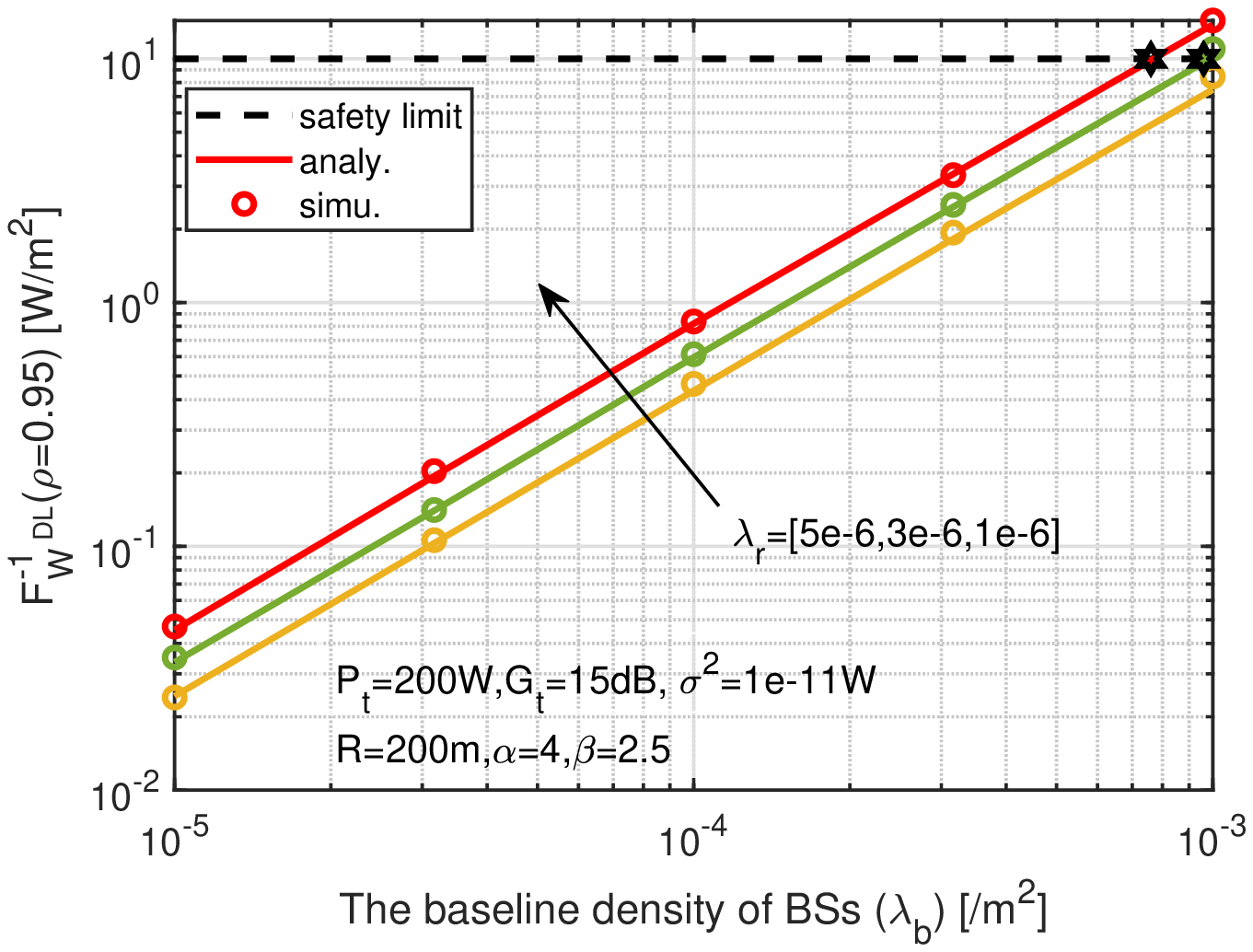}
    \caption{The $95$-th percentile of exposure level in the downlink outside the restricted area vs the baseline density of the \acp{BS}, for various densities of the restricted areas.}
    \label{fig:DL_density_inverse_EMF}
\end{minipage}
\hfill
\begin{minipage}{.49\textwidth}
    \centering
    \includegraphics[width=1\linewidth]{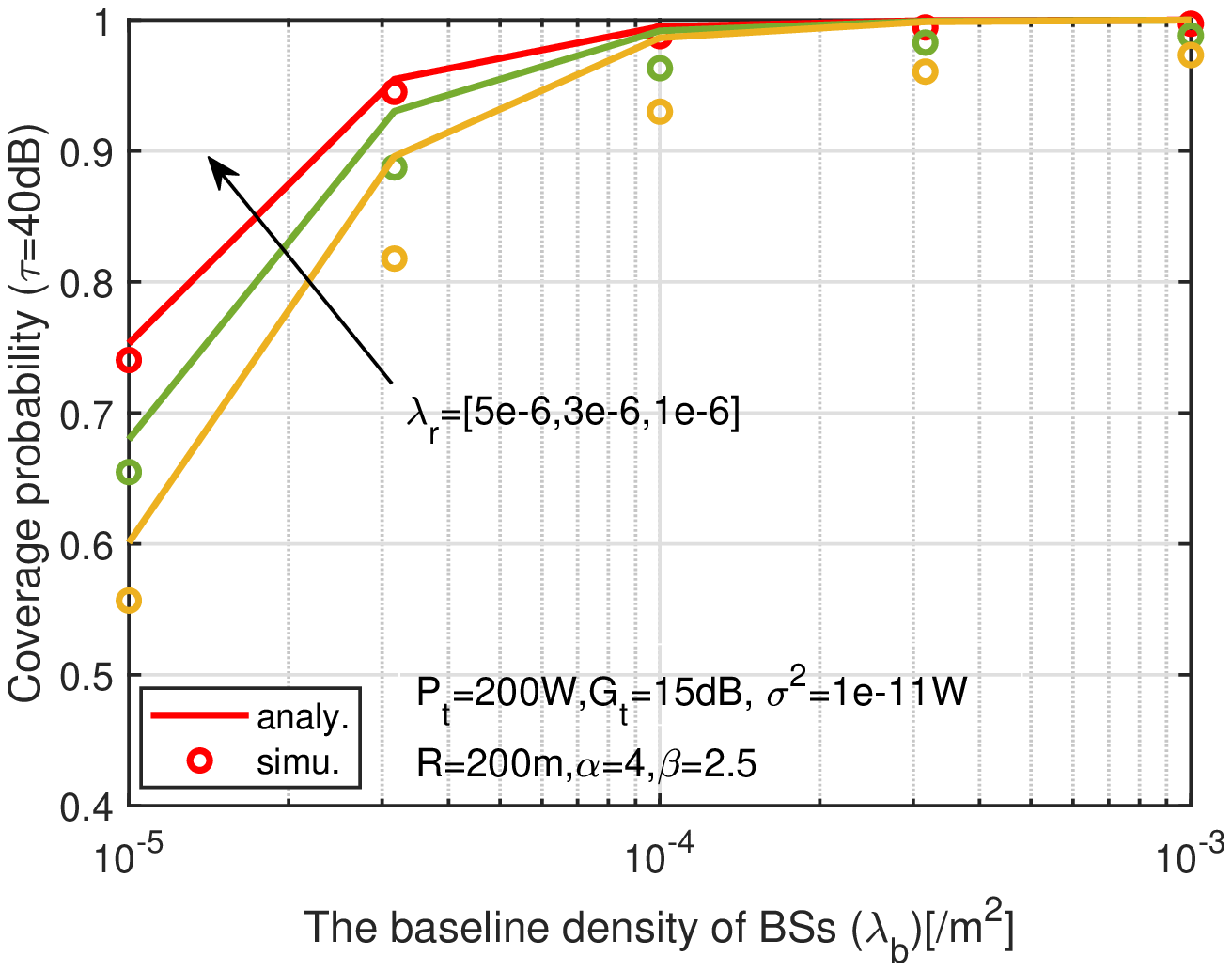}
    \caption{The \ac{CCDF} of the downlink \ac{SNR} levels outside the restricted area vs  baseline densities of \acp{BS}, for various densities of the restricted areas.}
    \label{fig:DL_density_SNR}
    \end{minipage}
\end{figure}
Fig.~\ref{fig:DL_density_inverse_EMF} and Fig.~\ref{fig:DL_density_SNR} show the impact of the density of the restricted areas/holes ($\lambda_r$) and the baseline density of the \acp{BS} ($\lambda_b$) on the \ac{EMF} exposure and coverage probability, respectively, for the users outside the restricted areas. Since the holes ensure a lower exposure level in the restricted areas than that outside the hole, the safety requirement of the user inside the hole can be guaranteed as long as the \ac{EMF} exposure of the user outside the hole is less than the maximum allowable value ($W_{\max}=10 ~\rm W/m^2$, depicted as the dashed horizontal line) defined by the \ac{FCC}~\cite{ITU5G:19}. 
We see that both the exposure and coverage probability decrease for higher values of $\lambda_r$, while densely deployed \acp{BS} favorably affect coverage probability but adversely affect \ac{EMF} exposure in the downlink.
The reason is that with the increasing \ac{BS} density, the serving \ac{BS} may be closer to the typical user, and thus \ac{SNR} would be improved. However, the dense deployment of \acp{BS} would increase the radiating sources, making the typical user more likely to be exposed to higher electromagnetic radiation level. In particular, the maximum density (identified by the star shaped markers in Fig.~\ref{fig:DL_density_inverse_EMF})  corresponds to the maximum allowable exposure. This maximum density is analyzed in the following paragraph.

In Fig.~\ref{fig:Maximum_density}, we depict the maximum allowable baseline density $\lambda_b^*$,  obtained by solving the optimization problem in (\ref{eq:optproblem1}), vs various maximum permitted \ac{EMF} exposure, for different values of $\lambda_r$. For example, considering the \ac{FCC} limit on the power density, $10 ~\rm W/m^2$, the optimal value of $\lambda_b$ maximizes the downlink coverage probability is $7.6\times10^{-4} ~\rm BSs/m^2$, for $\lambda_r=10^{-6}~ \rm holes/m^2$ and $R=200 ~\rm m$. 

\begin{figure}[t!]
\begin{minipage}{.49\textwidth}
    \centering
    \includegraphics[width=1\linewidth]{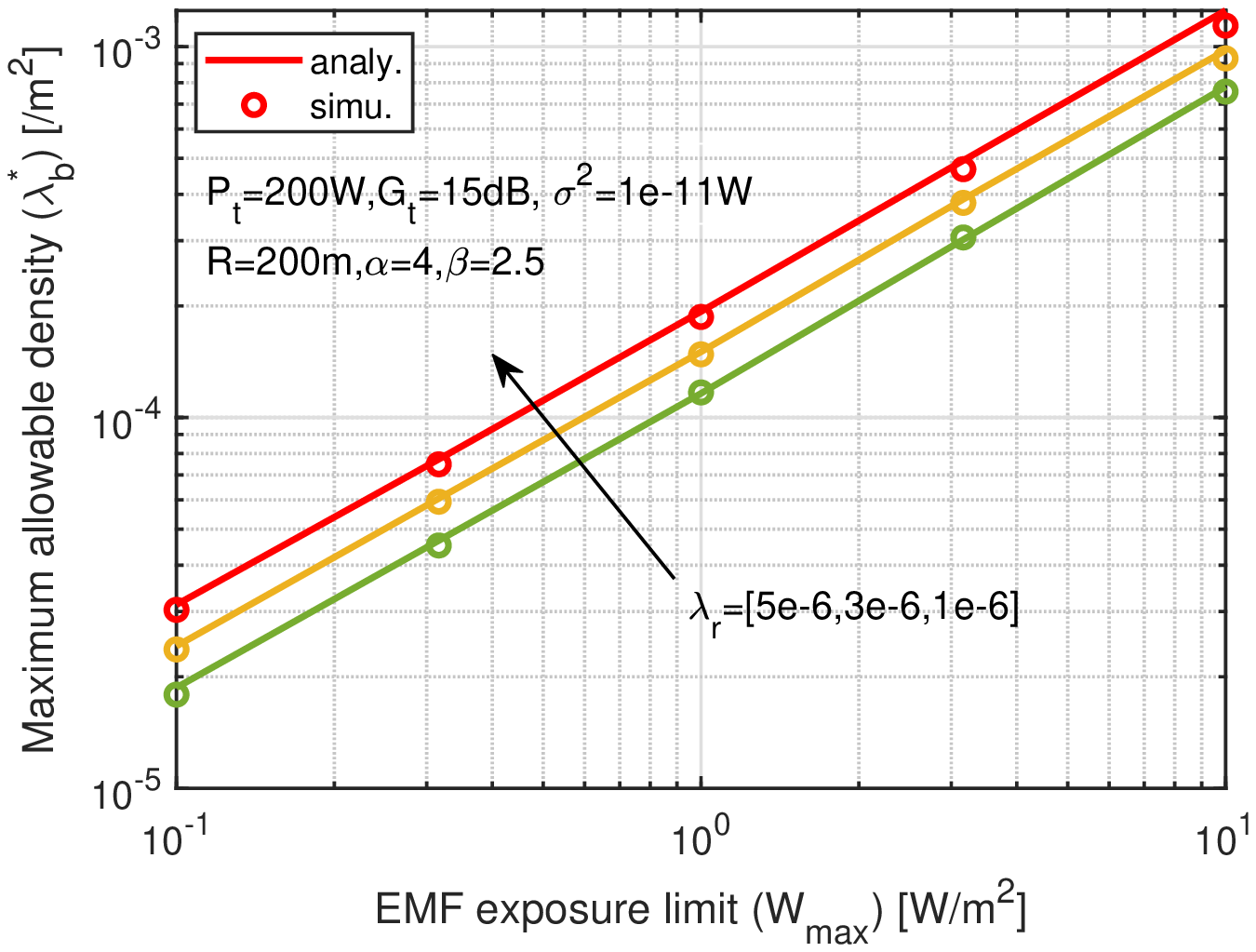}
    \caption{The maximum allowable baseline density $\lambda_b^*$ for different regulations on the maximum allowed \ac{EMF} exposure.}
    \label{fig:Maximum_density}
    \end{minipage}
\hfill
\begin{minipage}{.47\textwidth}
    \centering
    \includegraphics[width=1\linewidth]{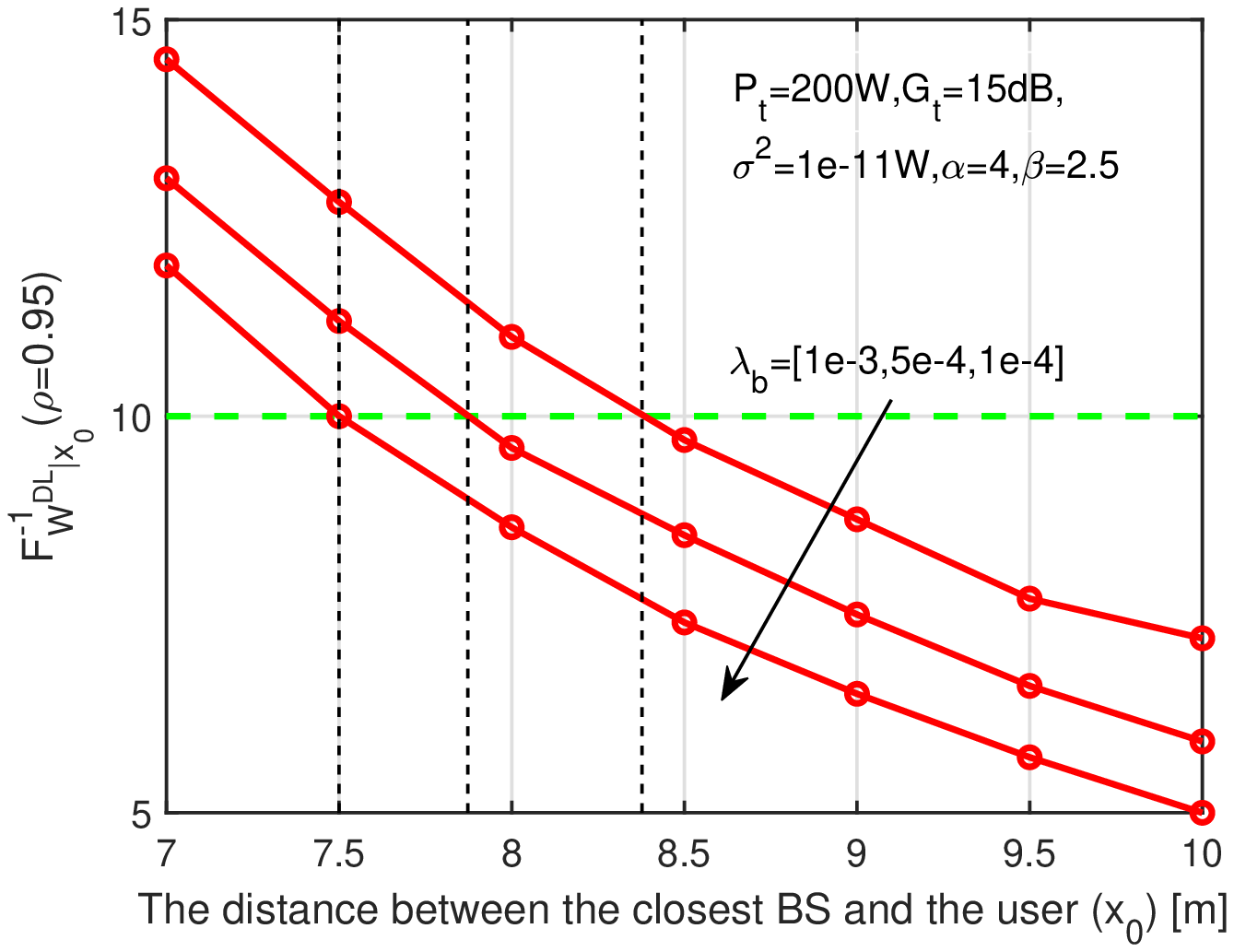}
    \caption{The $95$-th percentile of exposure level in the downlink conditioned on the closest \ac{BS} at $x_0$.}
    \label{fig:r0_EMF_DL}
    \end{minipage}
\end{figure}
In Fig.~\ref{fig:r0_EMF_DL}, we plot the $95$-th percentile of exposure level conditioned on the distance from the typical user to its closest \ac{BS} being at $x_0$ to calculate the compliance distance between users and \acp{BS}.
As mentioned in Sec.~\ref{subsec:Xcom}, compared with the FCC limits (depicted as the dashed horizontal line), the minimum compliance distance between a \ac{BS} and a user, $x_{\rm com}$, can be found. For the cellular network with $\lambda_b=10^{-4} ~\rm BSs/m^2$ and $\lambda_r=10^{-6} ~\rm holes/m^2$, $x_{\rm com}$ is $7.5~\rm m$, which satisfies $F_{W^{\rm DL}|7.5}(10)=0.95$. Namely, the public safety cannot be ensured if they enter an area centered at a \ac{BS} with a radius of $x_{\rm com}$.

\begin{figure}[t!]
\begin{minipage}{.49\textwidth}
    \centering
    \includegraphics[width=1\linewidth]{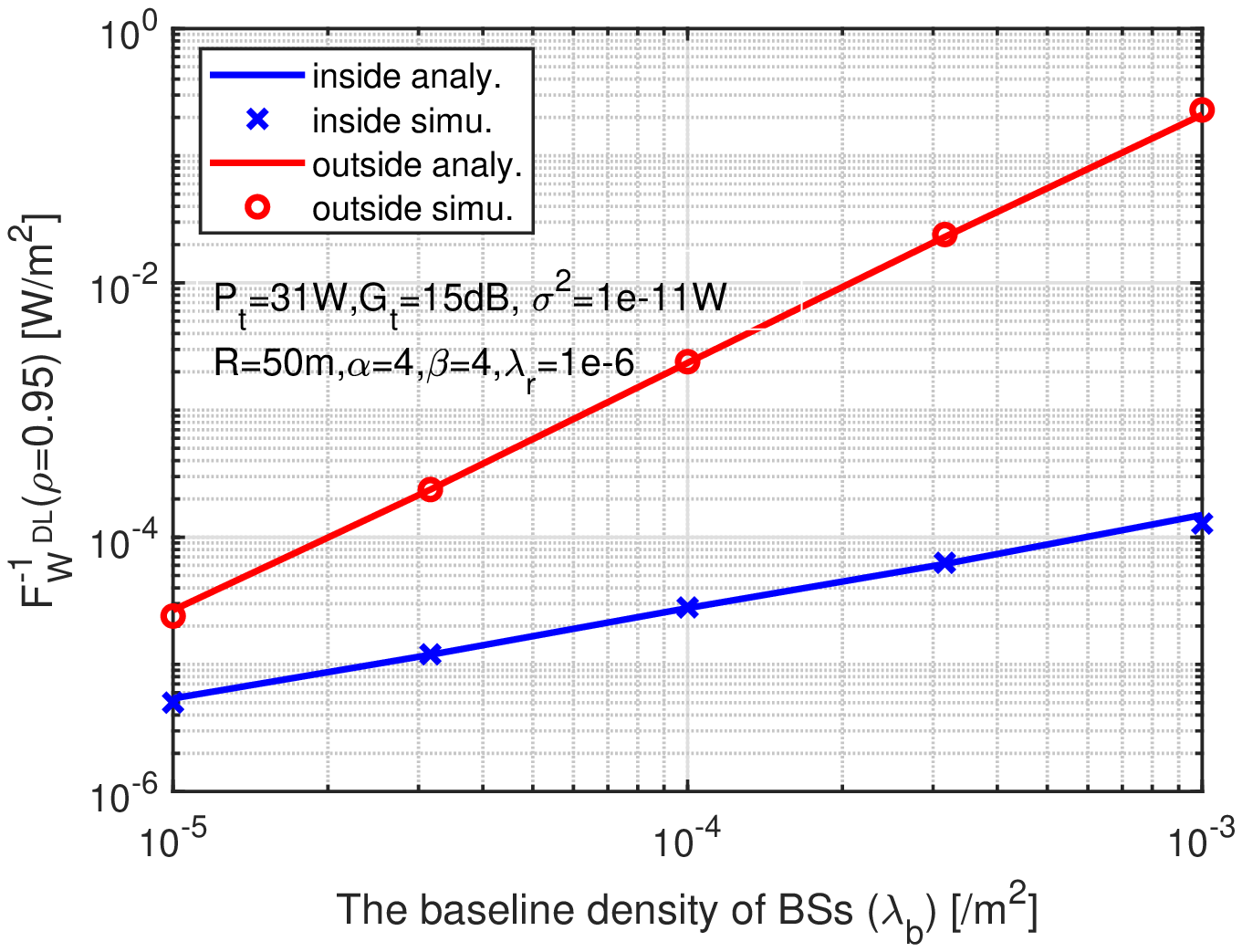}
    \caption{The $95$-th-percentile of exposure level in the downlink under the practical setting.}
    \label{fig:practicalEMF}
    \end{minipage}
\hfill
\begin{minipage}{.49\textwidth}
    \centering
    \includegraphics[width=1\linewidth]{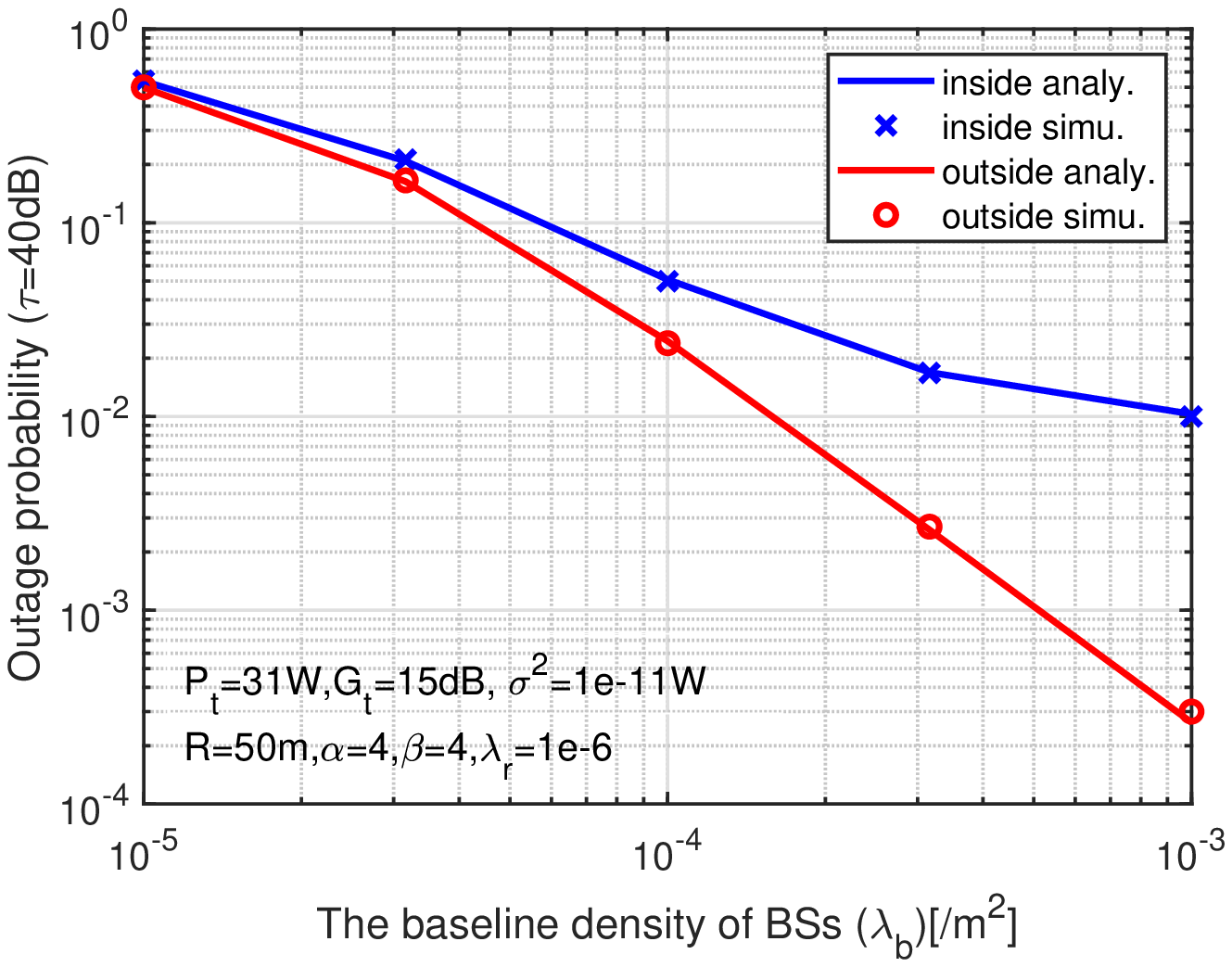}
    \caption{Compared with $\tau=40 \rm dB$, the outage probability in the downlink under the practical setting.}
    \label{fig:practicalOutage}
    \end{minipage}
\end{figure}
\subsection{Typical-Case Scenario for Downlink Exposure}\label{sec:simuDLtypical}
The results in a more common case are given in Fig.~\ref{fig:practicalEMF} and Fig.~\ref{fig:practicalOutage}, {where a statistical reduction factor of $0.31$ is applied to the general downlink transmit power ($100~\rm W$), i.e. $P_t=100~{\rm W}\times0.31$~\cite[Table 6]{chiaraviglio2021health}, and the same path-loss exponent ($\alpha=\beta=4$) is used for evaluating the network performance in the downlink.}  
Fig.~\ref{fig:practicalEMF} suggests that the $95$-th percentile of exposure level in such a typical case is far below the \ac{FCC} limit even in a network with the intensive deployment of \acp{BS}. This observation supports the conception that the development of network densification will not trigger severe health problems. In Fig.~\ref{fig:practicalOutage}, instead of giving the coverage probability, we plot the outage probability, the probability of an outage event when the \ac{SNR} cannot meet the required threshold, and we see that the simulation results closely match the analysis results.

\subsection{Uplink Exposure}\label{sec:simuUL}
%
In the uplink, fractional power control at the \ac{UE} is considered which makes the transmit power a function of the distance to the associated \ac{BS}. The maximum transmit power at the \ac{UE} is set to $200~\rm mW$ (i.e., the maximum transmit power in \ac{5G} mobile equipment\cite[Table 6]{chiaraviglio2021health}). 
{The distance between the user and its personal mobile equipment is assumed to be in far field of the \ac{UE} antenna with $u_0=20~\rm cm$. \footnote{{The far field of the antenna is defined as $u>\frac{2D^2}{c/f}$, where $D$ is the diameter of the antenna and $c/f$ is the wavelength of the uplink frequency $f$. Considering the \ac{LTE} using a frequency band of $2600 ~\rm MHz$,  $D=10~\rm cm$, and $c/f=(3\times 10^8 ~\rm m/s)/(2600 ~\rm MHz)$, we have $u>17.3~\rm cm$. Thus, $u_0=20\rm cm$ satisfies the far-field condition of the \ac{UE}'s antenna~\cite{sar}.}} 
Since the close distance between the user and its personal mobile equipment, the path-loss exponent ($\beta$) between them is typically smaller than that ($\alpha$) between the mobile equipment and its serving \ac{BS}. Thus, we set $\alpha=4$ and $\beta=2.5$.}
Besides, lower noise power is considered since the frequency band in uplink is normally narrower than that in downlink. 

\begin{figure}[t!]
\begin{minipage}{.49\textwidth}
    \centering
    \includegraphics[width=1\linewidth]{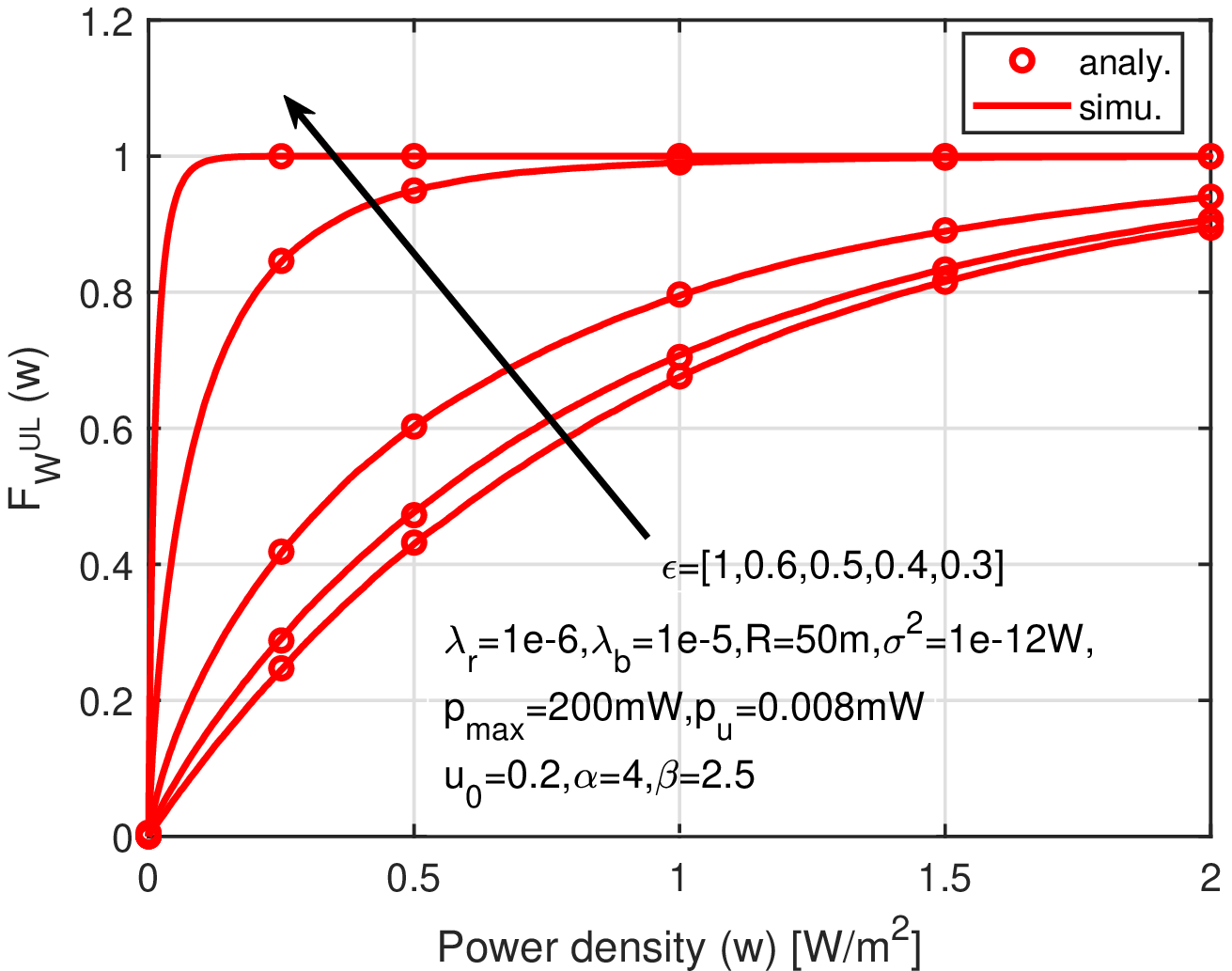}
    \caption{The \ac{CDF} of uplink power density levels outside the restricted area for various power control factors.}
\label{fig:epsilon_EMF_UL}
\end{minipage}
\hfill
\begin{minipage}{.49\textwidth}
    \centering
    \includegraphics[width=1\linewidth]{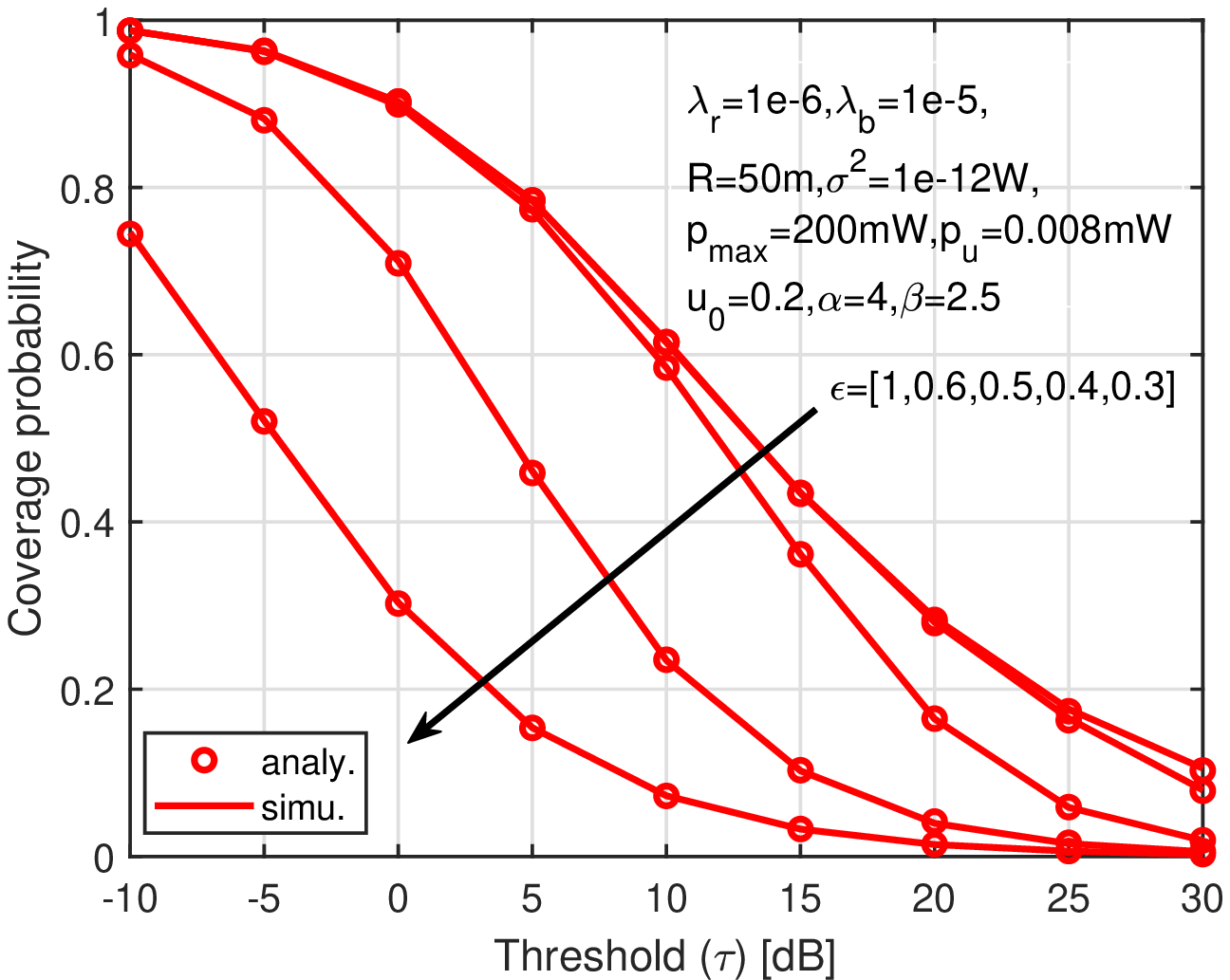}
    \caption{The \ac{CCDF} of the uplink \ac{SNR} levels outside the restricted area for various power control factors.}
    \label{fig:epsilon_SNR_UL}
    \end{minipage}
\end{figure}
The impact of power control factor $\epsilon$ on the distribution of uplink \ac{EMF} exposure and \ac{SNR} is presented in Fig.~\ref{fig:epsilon_EMF_UL} and Fig.~\ref{fig:epsilon_SNR_UL}, respectively. It can be seen from Fig.~\ref{fig:epsilon_EMF_UL} that increasing the power control factor $\epsilon$ leads to more severe exposure in uplink. This is because the transmit power of \ac{UE} with the distance-proportional power control increases with $\epsilon$ , as shown in (\ref{eq:PtranUL}).
Similarly, in Fig.~\ref{fig:epsilon_SNR_UL}, the larger the power control factor, the higher power the serving \ac{BS} receives, thereby enhancing the coverage performance.
Intuitively, as we increase $\epsilon$ up to $1$, the transmit power can completely compensate the path loss in the \ac{BS}-\ac{UE} link and the received power will equal to $p_{u}$ if the distance between a \ac{BS} and its serving user ($x_0$) is below $X_{\max}$ (as described in Sec.~\ref{sec:probelmUL}). 
However, if $x_0$ exceeds $X_{\max}$, then $P^{\rm UL}_{\rm tran}$ reaches its maximum value of $200~\rm mW$, which explains the small difference in Fig.~\ref{fig:epsilon_EMF_UL} and Fig.~\ref{fig:epsilon_SNR_UL} when changing $\epsilon$ from $0.6$ to $\epsilon=1$.


\begin{figure}[t!]
\begin{minipage}{.49\textwidth}
    \centering
    \includegraphics[width=1\linewidth]{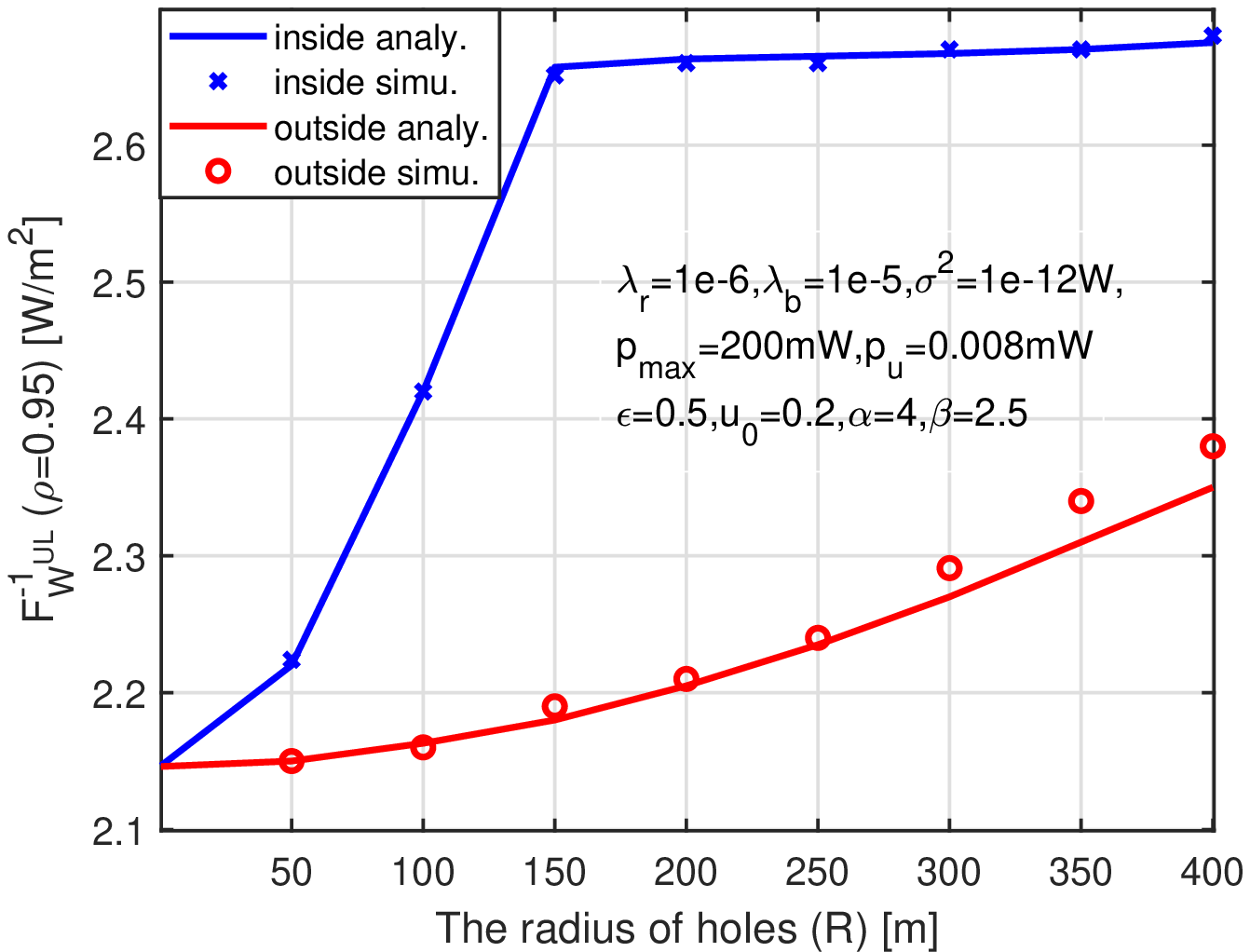}
    \caption{The $95$-th-percentile of exposure level in the uplink inside or outside the restricted area for various exclusion zone radii.}
\label{fig:Rholes_EMF_UL}
\end{minipage}
\hfill
\begin{minipage}{.49\textwidth}
    \centering
    \includegraphics[width=1\linewidth]{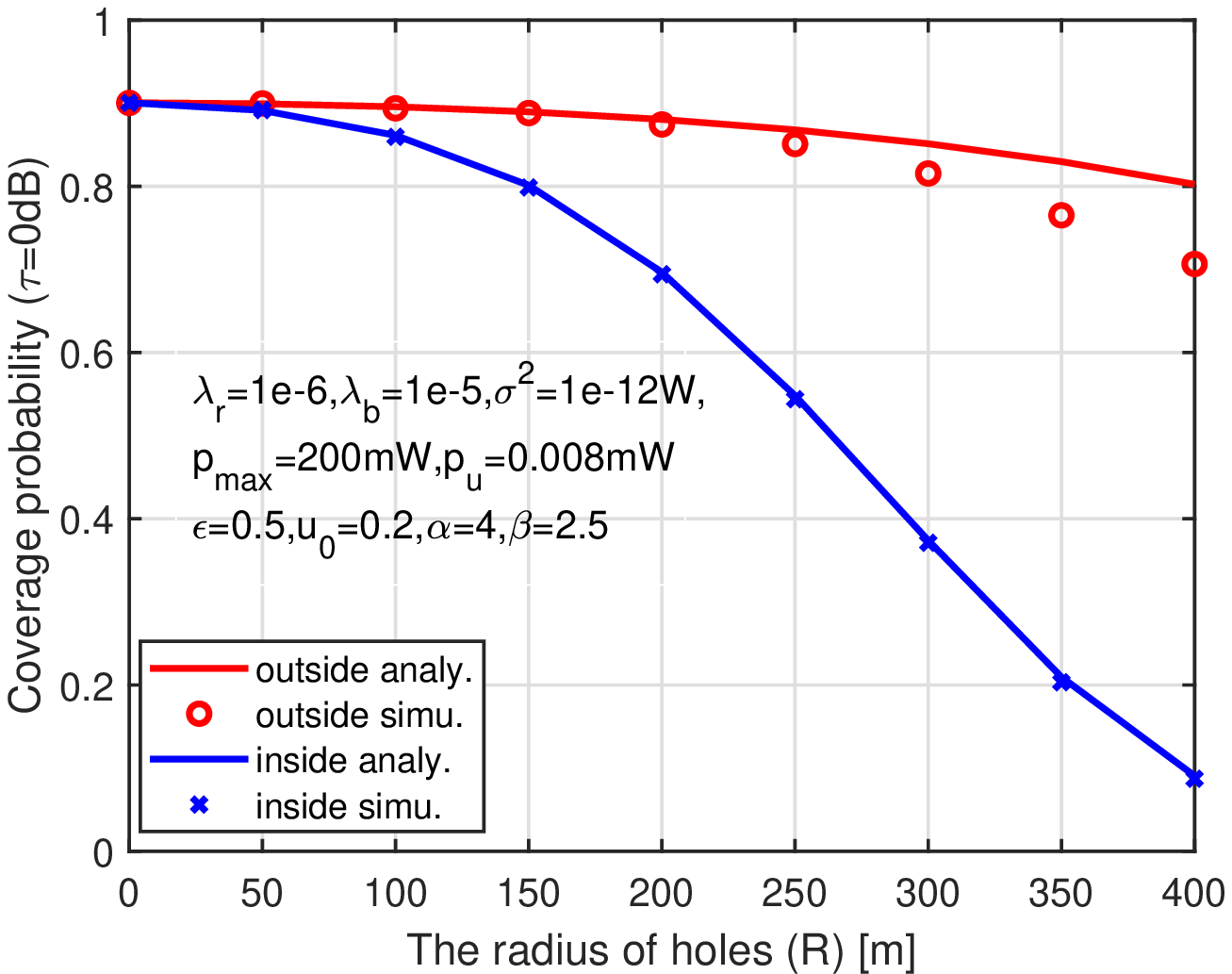}
    \caption{The \ac{CCDF} of the uplink \ac{SNR} levels inside or outside the restricted area for various exclusion zone radii.}
    \label{fig:Rholes_SNR_UL}
    \end{minipage}
\end{figure}
In Fig.~\ref{fig:Rholes_EMF_UL} and Fig.~\ref{fig:Rholes_SNR_UL}, we explore the effect of hole radius $R$ on the exposure level and the coverage probability in the uplink.
Different from simulation results on the downlink exposure in Sec.~\ref{sec:simuDL}, in Fig.~\ref{fig:Rholes_EMF_UL}, the user inside the hole is exposed to higher-level \ac{RFR} emitted from personal mobile equipment than the user outside the hole. Interestingly, we notice that increasing the zone radius $R$ does not mitigate \ac{RFR} for both two kinds of users (inside and outside the holes). 
The reason for this anomaly is that we control the transmit power at \ac{UE}, which is monotonically increasing with the distance between the user and its serving \ac{BS} as in (\ref{eq:PtranUL}). 
Particularly, for the user outside the hole, when we expand $R$, the density of \ac{PHP}-distributed \acp{BS}, $\lambda_B$, will decrease. The reduced number of \acp{BS} has no effect on uplink radiating sources since the uplink exposure is induced by an individual's mobile device. However, the distance between the closest \ac{BS} to the typical user may become farther, resulting in stronger transmit power at \ac{UE} and exposure to the user. 
We also notice that the uplink exposure level of the typical user inside the hole gradually tends to be constant when $R>150 ~\rm m$, which is the consequence of maximum transmit power constraint.
On the other hand, the uplink coverage probability decreases as we increase the hole radius, as shown in Fig.~\ref{fig:Rholes_SNR_UL}.  This trend is similar to the downlink scenario observed in Fig.~\ref{fig:Rholes_SNR_DL}.

%
\begin{figure}[t!]
\begin{minipage}{.49\textwidth}
    \centering
    \includegraphics[width=1\linewidth]{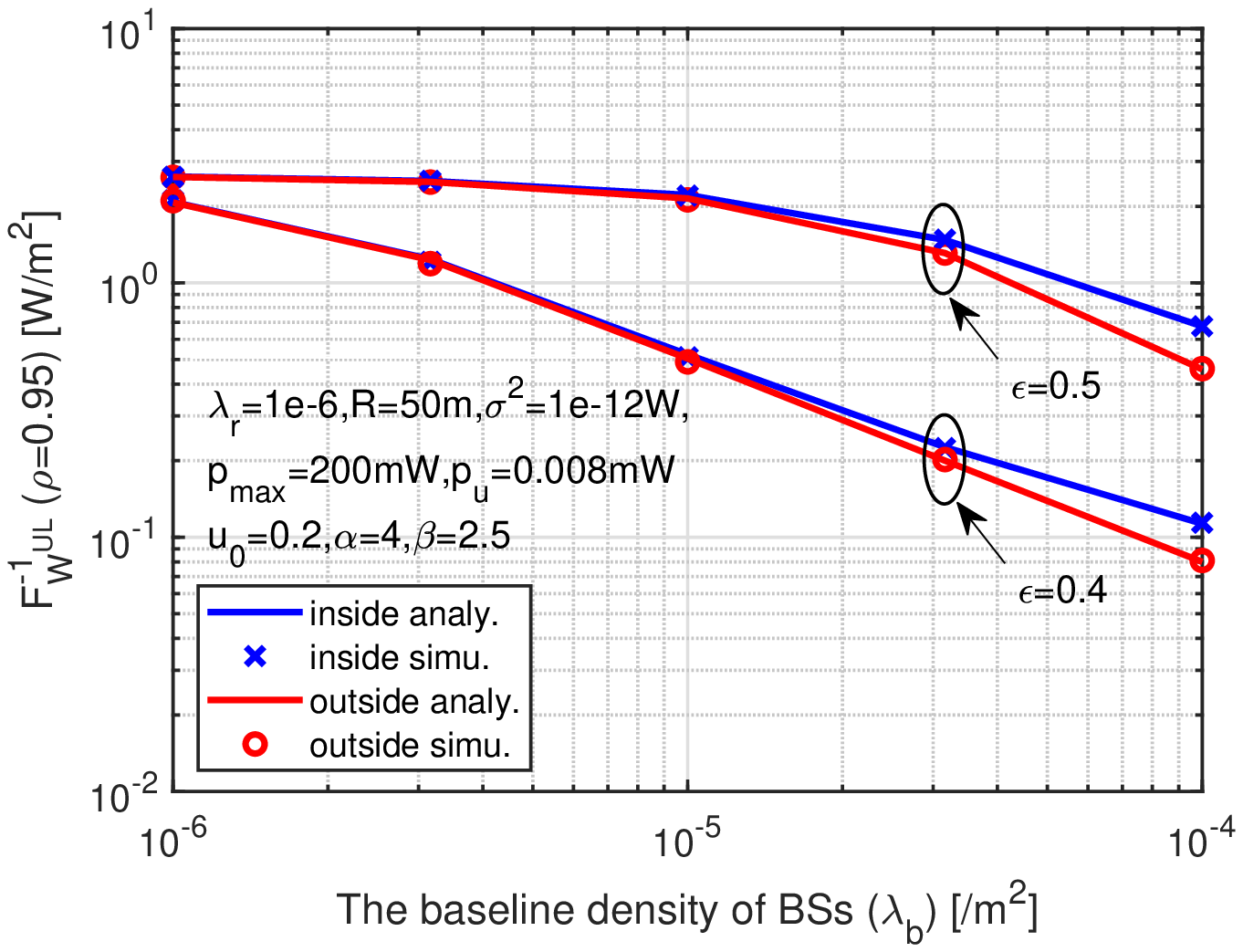}
    \caption{The $95$-th percentile of exposure level in the uplink inside or outside the restricted area vs various baseline densities of \acp{BS} for different values of power control factor.}
\label{fig:UL_density_inverse_EMF}
\end{minipage}
\hfill
\begin{minipage}{.49\textwidth}
    \centering
    \includegraphics[width=1\linewidth]{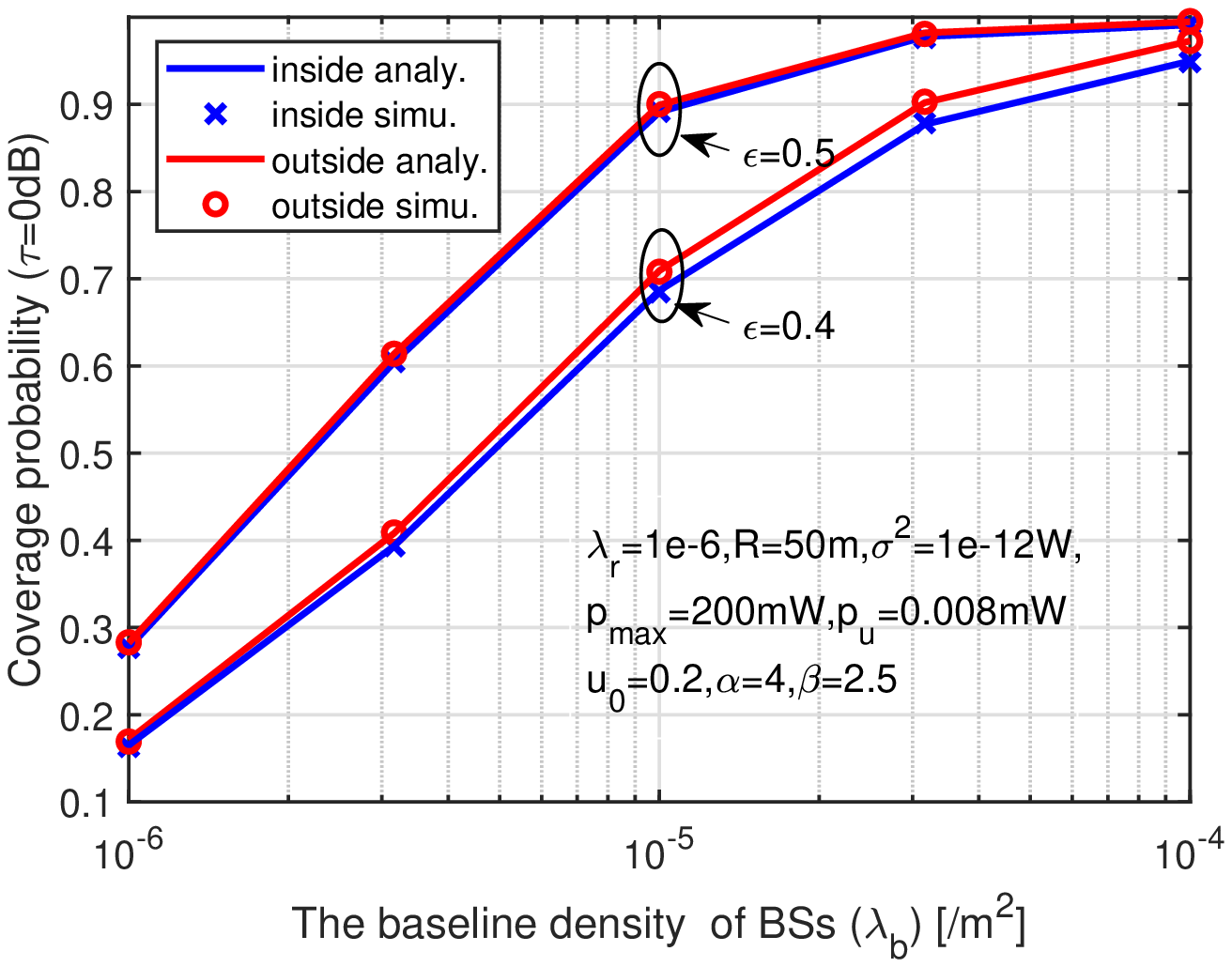}
    \caption{The \ac{CCDF} of the uplink \ac{SNR} levels inside or outside the restricted area vs various baseline densities of \acp{BS} for different values of power control factor.}
    \label{fig:UL_density_SNR}
    \end{minipage}
\end{figure}
Next, we discuss the influence of $\lambda_b$ on the network performance metrics in Fig.~\ref{fig:UL_density_inverse_EMF} and Fig.~\ref{fig:UL_density_SNR}. 
Increasing the value of $\lambda_b$ leads to a decrease in the $95$-th percentile of \ac{EMF} exposure (as shown in Fig.~\ref{fig:UL_density_inverse_EMF}) and an increase in coverage probability (as shown in Fig.~\ref{fig:UL_density_SNR}) for the typical user inside and outside the restricted area.
In fact, after increasing the baseline density $\lambda_b$, there are more \acp{BS} around the typical user, which has a potential to reduce the distance ($x_0$) between the typical user and its serving \ac{BS}. Meanwhile, the shorter distance $x_0$ leads to lower transmit power in (\ref{eq:PtranUL}) and lower \ac{EMF} exposure levels in (\ref{metrics:WUL}).
The improvement in coverage probability and the mitigation in \ac{EMF} exposure in uplink reveals that dense deployment of \acp{BS} is conducive to the future cellular network design. 


\subsection{Exposure Index (Joint Downlink\&Uplink Exposure)}
%
Generally, the \ac{SAR} value of voice usage is larger than that of other usage, such as data downloading/uploading. Considering such a worst case, we choose $\rm SAR^{UL}=0.0053$ and $\rm SAR^{DL}=0.0042$~\cite[Table 27]{sar}.

\begin{figure}[t!]
\begin{minipage}{.49\textwidth}
    \centering
    \includegraphics[width=1\linewidth]{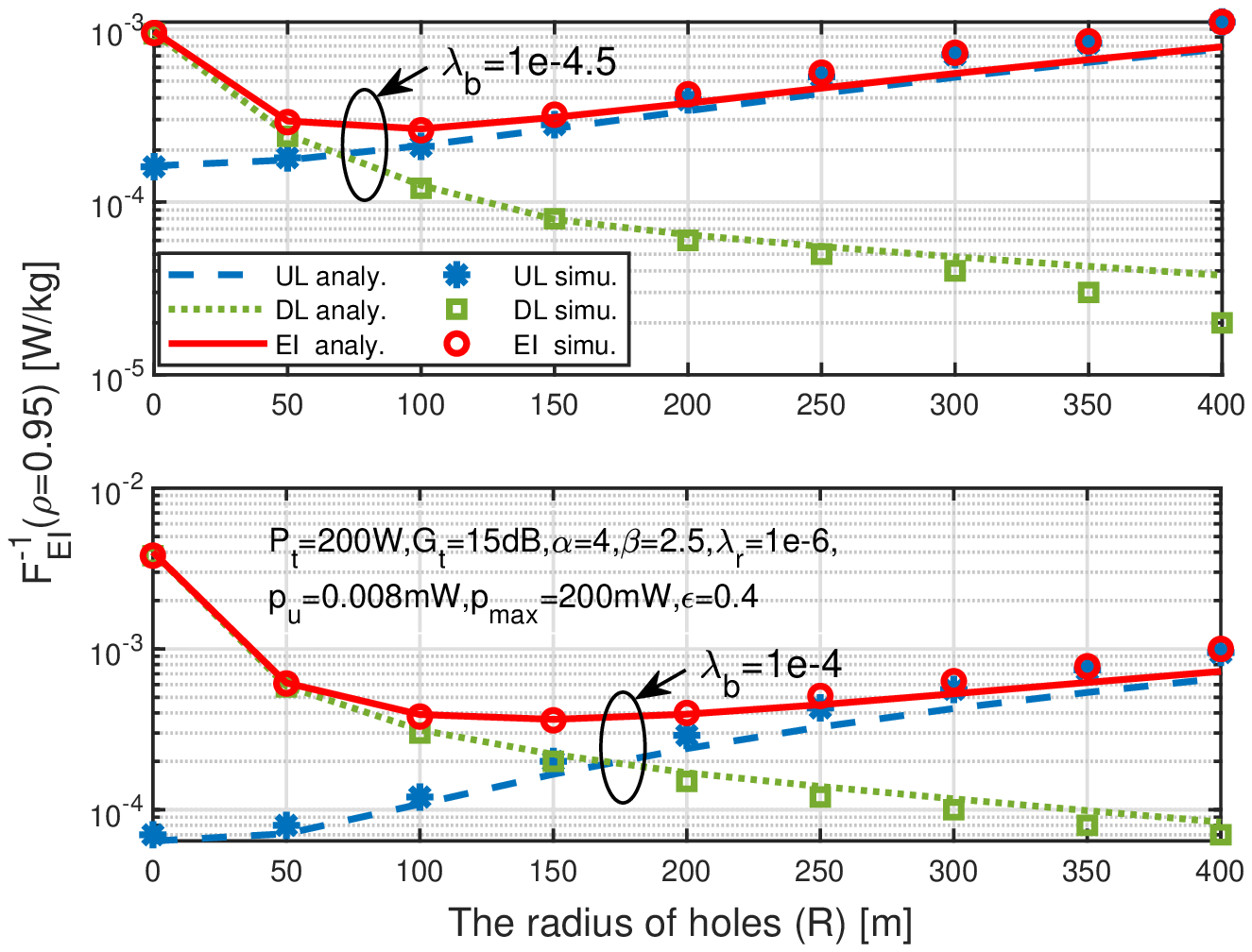}
    \caption{The $95$-th percentile of exposure level of $\rm EI$, $\rm EI^{UL}$ and $\rm EI^{DL}$ inside the restricted area vs various exclusion zone radii for different $\lambda_b$.}
    \label{fig:Rholes_EI}
    \end{minipage}
\hfill
\begin{minipage}{.49\textwidth}
    \centering
    \includegraphics[width=1\linewidth]{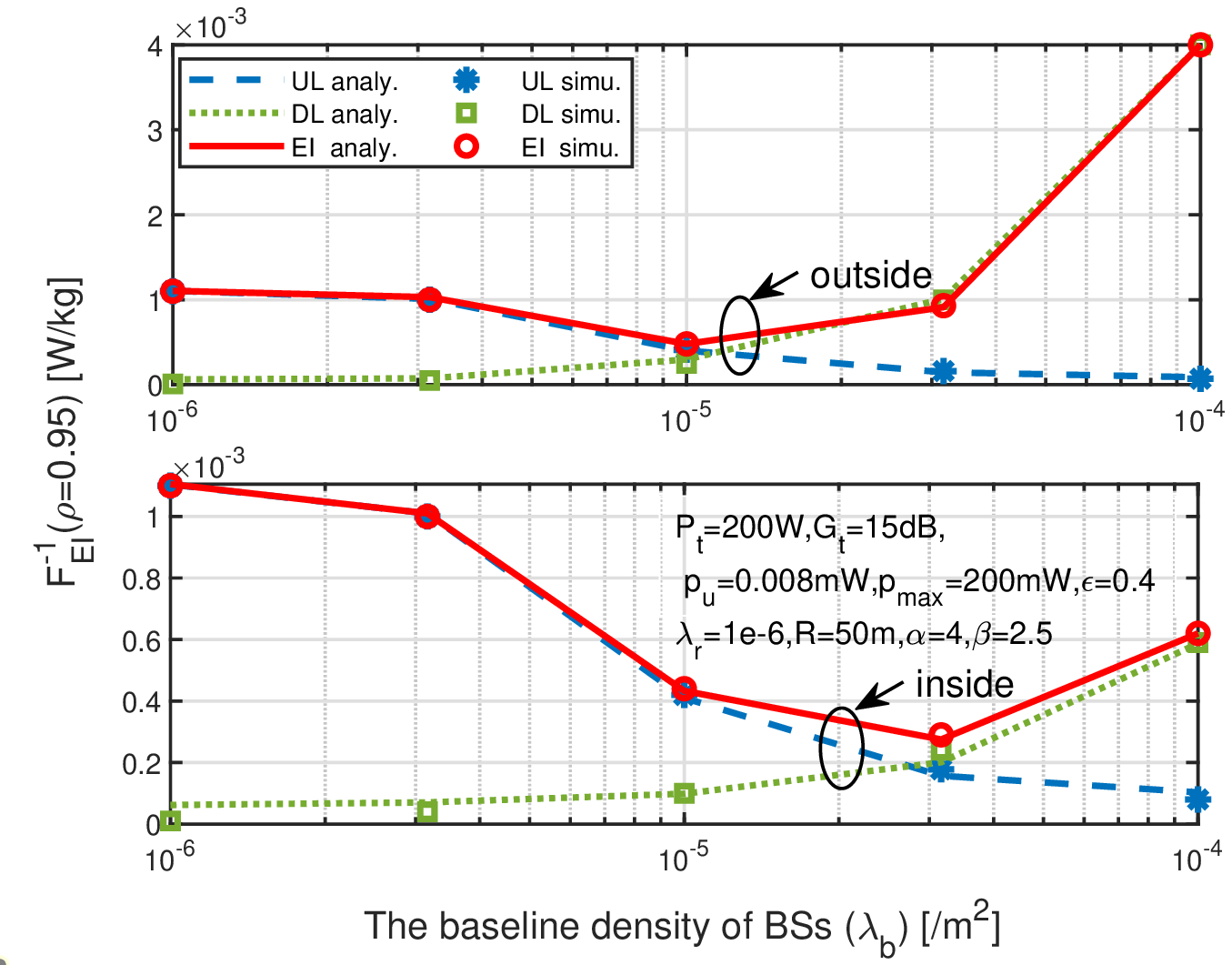}
    \caption{The $95$-th percentile of exposure level of  $\rm EI$, $\rm EI^{UL}$ and $\rm EI^{DL}$ inside or outside the restricted area vs various baseline densities of \acp{BS}.}
\label{fig:density_EI}
\end{minipage}
\end{figure}

Fig.~\ref{fig:Rholes_EI} reveals the impact of hole radius on the joint downlink\&uplink exposure.
It can be concluded from Fig.~\ref{fig:Rholes_EMF_DL} and Fig.~\ref{fig:Rholes_EMF_UL} that the existence of holes around the restricted areas cannot mitigate the exposure from \ac{UE} but it is effective for reducing the exposure from \acp{BS}.
The contradicting trend (between the increase in uplink exposure and the decrease in downlink exposure when expanding the exclusion zone radius)
reminds us that we cannot blindly protect the users by removing \acp{BS} near the restricted areas, which also causes both uplink and downlink coverage performance degradation as can be seen in Fig.~\ref{fig:Rholes_SNR_DL} and Fig.~\ref{fig:Rholes_SNR_UL}.
Therefore, Fig.~\ref{fig:Rholes_EI} considers joint downlink and uplink exposure as $\rm EI$ when resizing the hole radius.
In Fig.~\ref{fig:Rholes_EI}, when $R< 100~\rm m$, the total exposure is mainly from the \acp{BS} in downlink but when $R\ge 100~\rm m$, downlink exposure is gradually decreasing and uplink exposure becomes dominant. Namely, there exists an optimal value, e.g., $R^*=100~\rm m$, that minimizes the total exposure for the network with $\lambda_b=10^{-4.5} ~\rm BSs/m^2$.

In Fig.~\ref{fig:density_EI}, we plot the $95$-th percentile of $\rm EI$ for the typical user inside and outside the restricted area under different baseline densities of \acp{BS}.
The dense deployment of \acp{BS} can improve both uplink and downlink coverage probability as shown in Fig.~\ref{fig:DL_density_SNR} and Fig.~\ref{fig:UL_density_SNR}. Nevertheless, the uplink and downlink exposure levels show opposite trends with the increase of the baseline density $\lambda_b$, as shown in Fig.~\ref{fig:DL_density_inverse_EMF} and Fig.~\ref{fig:UL_density_inverse_EMF}, respectively. These imply that there is an optimal value of $\lambda_b$ that minimizes the $\rm EI$, i.e., a solution to the optimization problem  in \eqref{eq:optproblem3}.
For the typical user inside the hole, it can be observed from Fig.~\ref{fig:density_EI} that there is a turning point $\lambda_b^*=10^{-4.5} ~\rm BSs/m^2$. Before this point, $\rm EI$ is dominated by $\rm EI^{UL}$, and after it, $\rm EI$  is dominated by $\rm EI^{DL}$. At this point, the joint exposure, $\rm EI$, at the restricted area is minimized. For the typical user outside the hole, the optimal baseline density, $\lambda_b^*$, corresponding to the minimum $\rm EI$ is $10^{-5} ~\rm BSs/m^2$.



\section{Conclusion}
\label{sec:conc}

This paper integrated the EMF restrictions on the coverage performance and exposure analysis and formulated optimization problems on how to design the EMF-aware cellular networks.
Particularly, the distribution of \acp{BS} was generated by a \ac{PHP}, accounting for the distance between BSs and restricted areas where the presence of \acp{BS} is prohibited. 
Using tools of stochastic geometry, we analyzed the radiation and coverage probability in terms of downlink and uplink. Furthermore, we investigated the effect of system parameters on the joint downlink\&uplink radiation from both BSs and UE through {$\rm EI$}.
\textcolor{black}{
With the aid of numerical results, 
we showed that even the conservative evaluation of the $95$-th percentile of \ac{EMF} exposure level can still comply with the international guidelines, and the exposure in more typical settings is far below the maximum permissible level.
It can also be seen that increasing the baseline density of \acp{BS} or decreasing the permitted distance around restricted areas can reduce the exposure from mobile equipment in uplink while exacerbating the exposure from \acp{BS} in downlink. Such opposite trend demonstrated the reasonability of taking joint downlink\&uplink exposure into account when designing the system parameters for the \ac{EMF}-aware cellular network. We found that there exists optimal values of the distance between restricted areas and BSs and the baseline density of BSs that minimizes the total exposure under a certain network configuration.
}


%


%

\appendices
\section{Proof of Lemma~\ref{lemma:ClosestDL}}\label{app:lemma1}
The \ac{PHP} $\Psi_B$ can be approximated as a \ac{PPP} with density $\lambda_{B}={\lambda_{b}} \exp(-\lambda_{r}R^2)$.
Following the standard result of \ac{PPP}~\cite{stochasticgeometry}, the Lebesgue measure of the area centered at the typical user outside the hole can be expressed as
$\rho_{\rm out}(x)=\pi x^2,x \ge 0$.
%
%
Using the null probability of \ac{PPP} in~\cite{nullprob}, the \ac{CDF} of $X_{\rm out}$ is given by
\begin{equation}\label{eq:FXout}
\begin{split}
F_{X_{\rm out}}(x)
&=\mathbb{P}\left \{ X_{\rm out}\le x \right \} 
=1-\mathbb{P}\left \{ X_{\rm out} > x \right \}
=1-\exp\left [ \lambda_{B} \rho_{\rm out}(x) \right ]\\
&=1-\exp (-\lambda_{B}\pi x^2) 
=1-\exp \left [ -{\lambda_{b}} \exp(-\lambda_{r}R^2)\pi x^2 \right ].  
\end{split}
\end{equation}
Because of the relationship between \ac{PDF} and \ac{CDF}, i.e., $f_{X_{\rm out}}(x)=\frac{\mathrm{d} }{\mathrm{d} x} F_{X_{\rm out}}(x)$, we can obtain the \ac{PDF} of $X_{\rm out}$, $f_{X_{\rm out}}(x)$, in \eqref{eq:fXout}. 
For the typical user inside the hole, the Lebesgue measure of the area is given by
$\rho_{\rm in}(x)=\pi (x^2-R^2),x \ge R$,
%
%
where $R$ is the radius of the holes. Then, using the same steps of deriving $f_{X_{\rm out}}(x)$, 
we finish the proof of $f_{X_{\rm in}}(x)$.

\section{Proof of Theorem~\ref{theor:covDL}}\label{app:theorem1}
We start by analyzing the downlink coverage probability of the typical user outside the hole. With the distribution of $X_{\rm out}$ in (\ref{eq:fXout}), (\ref{metrics:snrDL}) and (\ref{metrics:covDL}) can be further processed as 
\begin{equation}\label{eq:covout1}
\begin{split}
&\mathcal{P}^{\rm DL}_{\rm cov}=\mathbb{P}\left \{  {\rm SNR^{DL}}>\tau \right \}
 =\mathbb{E}_{X_{\rm out}}\left [ \mathbb{P}\left \{  {\rm SNR^{DL}}>\tau | X_{\rm out}\right \} \right ] \\
&=\int_{0}^{\infty} \mathbb{P}\left \{ {\rm SNR}>\tau | X_{\rm out}=x_0\right \}f_{X_{\rm out}}(x_0)\,\mathrm{d}x_0
=\int_{0}^{\infty} \mathbb{P}\left \{ \frac{p \eta  H_0}{x_0^{\alpha}\sigma^2}>\tau \right \}f_{X_{\rm out}}(x_0)\,\mathrm{d}x_0.
\end{split}
\end{equation}
From the distribution of $H_0$ in \eqref{eq:ssfading}, we have 
\textcolor{black}{
\begin{equation}\label{eq:covout2}
\begin{split}
\mathbb{P}\left \{ \frac{p \eta H_0}{ x_0^{\alpha}\sigma^2}>\tau \right \}
&=\mathbb{P}\left \{ H_0>\frac{\tau \sigma^2}{p \eta x_0^{-\alpha}} \right \} 
=\frac{\Gamma_u\left ( m,s^{\rm DL}\sigma^2\right )}{\Gamma\left ( m \right ) }
\overset{(a)}{=} \sum_{k=0}^{m-1} \frac{(s^{\rm DL}\sigma^{2})^k}{k!}\exp\left ( -s^{\rm DL}\sigma^2\right ),
\end{split}
\end{equation}
}
\textcolor{black}{where $\Gamma_u\left ( m,mg \right )=\int_{mg}^{\infty} t^{m-1}e^{-t}\,\mathrm{d}t$, $s^{\rm DL}=\frac{m\tau}{p \eta x_0^{-\alpha}}$
and (a) is from the definition $\frac{\Gamma_u\left ( m,g \right )}{\Gamma\left ( m \right )} =\exp(-g) {\textstyle \sum_{k=0}^{m-1}}\frac{g^k}{k!}$.}
Taking (\ref{eq:covout2}) into (\ref{eq:covout1}), we obtain the expression of $\mathcal{P}^{\rm DL}_{\rm cov}$ for the typical user outside the hole. 
The proof of $\mathcal{P}^{\rm DL}_{\rm cov}$ for the typical user inside the hole is similar to the above approach, which therefore is omitted here.

\section{Proof of Theorem~\ref{theor:emfDL}}\label{app:theorem2}
In order to find the \ac{CDF} of \ac{EMF} exposure in the downlink, we first derive the Laplace transform of $W^{\rm DL}$, which is given by
\begin{equation}\label{eq:LemfDL0}
\begin{split}
\mathcal{L} _{W^{\rm DL}}(s)
&=\mathbb{E}_{W^{\rm DL}}\left [  \exp(-s W^{\rm DL}) \right ] 
=\mathbb{E}_{W^{\rm DL}}\left [ \exp\left (-s \sum_{i,b_i\in \Psi_B}\frac{p H_i}{4\pi x_i^{\beta}} \right )  \right ] 
\\&=\mathbb{E}_{\Psi_B}\left [ \prod_{i,b_i\in \Psi_B}\mathbb{E}_{H}\left [ \exp\left ( -s\frac{pH}{4\pi x_i^{\beta}} \right ) \right ]    \right ] 
\overset{(a)}{=} \textcolor{black}{\mathbb{E}_{\Psi_B}\left [ \prod_{i,b_i\in \Psi_B} \kappa^{\rm DL}(x_i,s)  \right ],}
\end{split}
\end{equation}
where $\kappa^{\rm DL}(x_i,s)=\left(\frac{m}{m+sp(4\pi)^{-1}x_i^{-\beta}}\right )^m$ and (a) is from 
the distribution of small-scale fading $H$ in \eqref{eq:ssfading}.
Focusing on the typical user outside the hole and employing the probability generating functional (PGFL) of \ac{PPP} in~\cite{andrews2016primer}, (\ref{eq:LemfDL0}) can be further expressed as
\begin{equation}\label{eq:LemfDL1}
\begin{split}
\mathcal{L} _{W^{\rm DL}}(s)
&=\textcolor{black}{\exp\left ( -2\pi\lambda_B\int_{0}^{\infty} \left [ 1- \kappa^{\rm DL}(x,s)  \right ]  x\mathrm{d}x \right )},
\end{split}
\end{equation}
where $\lambda_{B}={\lambda_{b}}\exp(-\lambda_{r}R^2)$.
While for the user inside the hole, $x_i\ge R$ and thus the Laplace transform of $W^{\rm DL}$ is 
\begin{equation}\label{eq:LemfDL2}
\begin{split}
\mathcal{L} _{W^{\rm DL}}(s)
&=\textcolor{black}{\exp\left ( -2\pi\lambda_B\int_{R}^{\infty}\left [ 1-\kappa^{\rm DL}(x,s)  \right ]  x\mathrm{d}x \right )}.
\end{split}
\end{equation}
From Gil-Pelaez theorem, the \ac{CDF} of $W^{\rm DL}$ can be written as
\begin{align} \label{eq:CDF}
    F_{W^{\rm DL}}(w)&= \frac{1}{2}- \frac{1}{\pi} \int_{0}^{\infty} \frac{1}{t} \, \operatorname{Im} \left( \exp(-{j} \, t w) \phi_{W^{\rm DL}}(t)               \right) \mathrm{d}t \nonumber \\
  &= \frac{1}{2}- \frac{1}{\pi} \int_{0}^{\infty} \frac{1}{t} \, \operatorname{Im} \left( \exp(-{j}\, t w) \mathcal{L}_{W^{\rm DL}}(-{j}\, t)               \right) \mathrm{d}t \nonumber \\
  &= \frac{1}{2}- \frac{1}{2 {j} \,\pi} \int_{0}^{\infty} \frac{1}{t} \, \left[ e^{-{j}\, t w} \mathcal{L}_{W^{\rm DL}}(-{j}\, t) -   e^{{j}\, t w} \mathcal{L}_{W^{\rm DL}}({j}\, t)           \right] \mathrm{d}t,
\end{align}
where
$  \phi_{W^{\rm DL}}(t) =   \mathbb{E} \left\{\exp{\left( {j}\, t \, w \right)}  \right\}
  = \mathcal{L}_{W^{\rm DL}}(-{j} \,t) $.
Submitting (\ref{eq:LemfDL1}) or (\ref{eq:LemfDL2}) into (\ref{eq:CDF}), we prove Theorem~\ref{theor:emfDL}.

\section{Proof of Theorem~\ref{theor:emfx0}}\label{app:theorem3}
Under the condition that the serving \ac{BS} is located at $b_0$ with distance $x_0$ to the typical user, the power density received at the typical user is from the closest \ac{BS} with a distance of $x_0$ and from the rest of \acp{BS} with distances greater than $x_0$ to the typical user,
which is given by
\begin{equation}\label{eq:emfDLx0}
W^{\rm DL}(x_0)=\frac{pH_0}{4\pi x_0^{\beta}}+\sum_{i,b_i\in \Psi_B \setminus \{b_0\}}\frac{pH_i}{4\pi x_i^{\beta}}.
\end{equation}
The Laplace transform of $W^{\rm DL}(x_0)$ is given by
%
\begin{equation}\label{eq:LemfDLx0}
\begin{split}
\mathcal{L} _{W^{\rm DL}|x_0}(s)
&=\mathbb{E}_{W^{\rm DL}|{x_0}}\left [ \exp\left (-s \frac{pH_0}{4\pi x_0^{\beta}}-s\sum_{i,b_i\in \Psi_B \setminus \{b_0\}}\frac{pH_i}{4\pi x_i^{\beta}} \right )  \right ] 
\\&=\mathbb{E}_{H_0}\left [ \exp(-s\frac{pH_0}{4\pi x_0^{\beta}}) \right ] 
\mathbb{E}_{\Psi_B,\left \{ H_i \right \}}\left [ \prod_{i,b_i\in \Psi_B \setminus \{b_0\}} \exp\left ( -s\frac{pH_i}{4\pi x_i^{\beta}} \right )   \right ] 
\\&\overset{(a)}{=}\textcolor{black}{\kappa^{\rm DL}(x_0,s)\mathbb{E}_{\Psi_B}\left [ \prod_{i,b_i\in \Psi_B \setminus \{b_0\}}\kappa^{\rm DL}(x_i,s)  \right ]},
\end{split}
\end{equation}
where $\kappa^{\rm DL}(\cdot)$ is given in \eqref{eq:LemfDL} and (a) is from the same method in \eqref{eq:LemfDL0}.
Applying the PGFL of \ac{PPP}~\cite{andrews2016primer} into (\ref{eq:LemfDLx0}), we have
\begin{equation}\label{eq:LemfDLx1}
\begin{split}
\mathcal{L} _{W^{\rm DL}|x_0}(s)
&{=}\textcolor{black}{\kappa^{\rm DL}(x_0,s)\exp\left ( -2\pi\lambda_B\int_{x_0}^{\infty}\left [ 1-\kappa^{\rm DL}(x,s)  \right ]  x\,\mathrm{d}x \right )},
\end{split}
\end{equation}
where $\lambda_{B}={\lambda_{b}} \exp(-\lambda_{r}R^2)$.
Following the result of (\ref{eq:CDF}), we complete the proof of Theorem~\ref{theor:emfx0}
by replacing $\mathcal{L} _{W^{\rm DL}}(s)$ with $\mathcal{L} _{W^{\rm DL}|x_0}(s)$.
\ifCLASSOPTIONcaptionsoff
  \newpage
\fi



\bibliographystyle{IEEEtran}
\bibliography{emf_ref}
%

%







\end{document}